%% file: arxiv.tex
\documentclass[acmsmall,10pt]{acmart}
\startPage{1}

\setcopyright{none}

\usepackage{booktabs}   
\usepackage{subcaption} 

\usepackage[T1]{fontenc}
\usepackage[utf8]{inputenc}

\usepackage{placeins}
\usepackage{subcaption}

\usepackage{amsmath}
\usepackage{amssymb}

\usepackage{graphicx}

\usepackage{xcolor}
\usepackage{hyperref}

\usepackage{bcprules}
\usepackage[shorthand]{semantic}
\mathlig{|->}{\mapsto}

\usepackage{changepage}
\usepackage{placeins}

\usepackage{color}
\definecolor{deepblue}{rgb}{0,0,0.5}
\definecolor{deepred}{rgb}{0.6,0,0}
\definecolor{deepgreen}{rgb}{0,0.5,0}

\usepackage{xcolor}
\definecolor{capabilitycolor}{RGB}{39, 121, 156}
\definecolor{capabilitybg}{RGB}{225, 239, 245}

\usepackage{listings}
\usepackage{markdown}

\usepackage{tikz}
\makeatletter
\newenvironment{btHighlight}[1][]
{\begingroup\tikzset{bt@Highlight@par/.style={#1}}\begin{lrbox}{\@tempboxa}}
{\end{lrbox}\bt@HL@box[bt@Highlight@par]{\@tempboxa}\endgroup}

\newcommand\btHL[1][]{%
  \begin{btHighlight}[#1]\bgroup\aftergroup\bt@HL@endenv%
}
\def\bt@HL@endenv{%
  \end{btHighlight}%
  \egroup
}
\newcommand{\bt@HL@box}[2][]{%
  \tikz[#1]{%
    \pgfpathrectangle{\pgfpoint{1pt}{0pt}}{\pgfpoint{\wd #2}{\ht #2}}%
    \pgfusepath{use as bounding box}%
    \node[anchor=base west, fill=light-gray,outer sep=0pt,inner xsep=1pt, inner ysep=0pt, rounded corners=3pt, minimum height=\ht\strutbox+1pt,#1]{\raisebox{1pt}{\strut}\strut\usebox{#2}};
  }%
}
\makeatother

\lstdefinelanguage{scalaish}{
  basicstyle=\smaller\ttfamily,
  keywords={erased, val, var, if, then, in, handle,
    return, def, match, case, new, type, trait,
     package, object, given, eff,
     pretype, class, extends, extension, infix, else,
     box, unbox, try, catch, import, throw, throws, using},
  keywordstyle=\bfseries,
  sensitive=false,
  comment=[l]{//},
  commentstyle=\color{gray},
  stringstyle=\color{gray}, 
  morestring=[b]',
  morestring=[b]",
  moredelim=**[is][\btHL]{`}{`},
}
\input{cc-macros}

\acmConference{Technical Report}{2021}{arXiv}
\acmYear{2022}
\acmISBN{} 
\acmDOI{} 
\startPage{1}

\setcopyright{none}

\begin{document}

\title{Scoped Capabilities for Polymorphic Effects}         


\author{Martin Odersky}
\affiliation{\institution{EPFL}}
\email{martin.odersky@epfl.ch}

\author{Aleksander Boruch-Gruszecki}
\affiliation{\institution{EPFL}}
\email{aleksander.boruch-gruszecki@epfl.ch}

\author{Edward Lee}
\affiliation{\institution{University of Waterloo}}
\email{e45lee@uwaterloo.ca}

\author{Jonathan Brachthäuser}
\affiliation{\institution{Eberhard Karls University of Tübingen}}
\email{jonathan.brachthaeuser@uni-tuebingen.de}

\author{Ondřej Lhoták}
\affiliation{\institution{University of Waterloo}}
\email{olhotak@uwaterloo.ca}


\begin{abstract}
  Type systems usually characterize the shape of values but not their
  free variables. However, many desirable safety properties could be
  guaranteed if one knew the free variables captured by values.
  We describe \CC, a calculus where such captured variables
  are succinctly represented in types, and show it can be used
  to safely implement effects and effect polymorphism via {\em scoped capabilities}.
  We discuss how the decision to track captured variables guides
  key aspects of the calculus, and show that \CC admits simple
  and intuitive types for common data structures and their typical
  usage patterns. We demonstrate how these ideas can be used to guide the
  implementation of capture checking in a practical programming language.


\end{abstract}

\maketitle

\renewcommand{\shortauthors}{Odersky, Boruch-Gruszecki, Lee, Brachthäuser, and Lhot\'ak}

\section{Introduction}\label{sec:introduction}

{\em Effects} are aspects of computation that go beyond describing shapes
of values and that we still want to track in types.
What exactly is modeled as an effect is a
question of language or library design. Some possibilities are:
{\it reading} or {\it writing} to mutable state outside a function,
{\it throwing an exception} to signal abnormal termination of a function,
{\it I/O} including file operations, network access, or user interaction,
{\it non-terminating} computations,
{\it suspending} a computation e.g., waiting for an event, or
{\it using a continuation} for control operations.

Despite hundreds of published papers
there is comparatively little adoption of static effect checking in programming languages.
The few designs that {\it are} widely implemented (for instance Java's checked exceptions or
monadic effects in some functional languages) are often critiqued for being both too verbose and too rigid.
The problem is not lack of expressiveness -- systems have been proposed and implemented for many
quite exotic kinds of effects. Rather, the problem is simple lack of usability
and flexibility, with particular difficulties in describing polymorphism. This leads either to overly complex definitions, or to the
requirement to duplicate large bodies of code.

Classical type-systematic approaches
fail since effects are inherently transitive along the edges of the dynamic call-graph:
A function's effects include the effects of all the functions it calls, transitively.
Traditional type and effect systems have no lightweight mechanism to describe this behavior.
The standard approach is either manual specialization along specific effect classes,
which means large-scale code duplication, or quantifiers on all definitions
along possible call graph edges to account for the possibility that some call target has an
effect, which means large amounts of boilerplate code. Arguably, it is this problem
more than any other that has so far hindered wide scale application of effect systems.

A promising alternative that circumvents this problem is to model effects
via capabilities \linebreak
\cite{marino09generic,gordon2020designing,liu2016,brachthaeuser2020effects,osvald2016gentrification}.
Capabilities exist in many forms, but we will restrict the meaning here to simple object
capabilities represented as regular program variables.
For instance, consider the following four formulations of a method in Scala\footnote{Scala 3.1 with
language import {\tt saferExceptions} enabled.} which are all morally equivalent:

\begin{lstlisting}[language=scalaish]
  def f(): T throws E
  def f(): throws[T, E]
  def f(): CanThrow[E] ?=> T
  def f()(using ct: CanThrow[E]): T
\end{lstlisting}

\noindent
The first version looks like it describes an effect: Function {\tt f} returns a {\tt T}, or it might throw exception {\tt E}.
The second version expresses this information in a type and the third version expands the type
to a kind of function type. These transitions are enabled by the following definition of type {\tt throws}:

\begin{lstlisting}[language=scalaish]
  infix type throws[T, E <: Exception] = CanThrow[E] ?=> T
\end{lstlisting}

\noindent
The type {\tt CanThrow[E] ?=> T} is a {\em context function type}. It represents functions
from {\tt CanThrow[E]} to {\tt T} that are applied implicitly to arguments synthesized by
the compiler. The definition is morally equivalent to the fourth version above, which
employs a {\tt using} clause to declare an implicit parameter {\tt ct}. The mechanics of context function types and their
expansion to implicit parameters was described by \citet{odersky2017simplicitly}.

In the fourth version, the parameter {\tt ct} acts as an (object) capability. A capability of
type {\tt CanThrow[E]} is required to perform the effect of throwing exception {\tt E}.

An important benefit of this switch from effects to capabilities is that it gives us
polymorphism for free. For instance, consider the {\tt map} function in
class {\tt List[A]}. If we wanted to take effects into account, it would look like this:

\begin{lstlisting}[language=scalaish]
  def map[B, E](f: A -> B eff E): List[B] eff E
\end{lstlisting}

\noindent
Here, {\tt A -> B eff E} is hypothetical syntax for the type of functions from {\tt A} to {\tt B}
that can have effect {\tt E}. While looking reasonable in the small, this
scheme quickly becomes unmanageable, if we consider that
every higher-order function has to be expanded that way and, furthermore, that
in an object-oriented language almost every method is a higher-order function
\cite{cook09understanding}. Indeed, many designers of programming languages with support for effect systems agree that
programmers should ideally not be confronted with explicit effect quantifiers
\citep{brachthaeuser2020effects,leijen2017type,lindleyBeBe2017}.

On the other hand, here is the type of {\tt map} if we represent effects with capabilities.
\begin{lstlisting}[language=scalaish]
  def map(f: A => B): List[B]
\end{lstlisting}
Interestingly, this is exactly the same as the type of {\tt map} in current Scala,
which does not track effects! In fact, compared to effect systems, we now decompose the space of possible effects differently:
{\tt map} is classified as pure since it does not produce any effects in its own code, but when analyzing an
application of {\tt map} to some argument, the capabilities required by the argument
are also capabilities required by the whole expression. In that sense, we get effect polymorphism for free.

The reason this works is that in an effects-as-capabilities discipline,
the type {\tt A => B} represents the type of {\em impure} function values
that can {\em close over} arbitrary effect capabilities. (Alongside, we also define a
type of pure functions {\tt A -> B} that are not allowed to close over
capabilities.)

This seems almost too good to be true, and indeed there is a catch: It now becomes
necessary to reason about capabilities captured in closures, ideally by
representing such knowledge in a type system.

\begin{figure}[h]

\begin{minipage}[t]{0.48\textwidth}\footnotesize
\begin{lstlisting}[language=scalaish]
    class TooLarge extends Exception

    def f(x: Int): Int throws TooLarge =
      if x < limit then x * x
      else throw TooLarge()

    val xs: List[Int]
    try xs.map(x => f(x))
    catch case TooLarge => Nil
\end{lstlisting}
\end{minipage}\begin{minipage}[t]{0.52\textwidth}\footnotesize
\begin{lstlisting}[language=scalaish]
      def f(x: Int)
            (using CanThrow[TooLarge]): Int =
        if x < limit then x * x
        else throw TooLarge())

      val xs: List[Int]
      try xs.map(x =>
            f(x)(using new CanThrow[TooLarge]))
      catch case TooLarge => Nil
\end{lstlisting}
\end{minipage}
\caption{Exception handling: source (left) and compiler-generated code (right)}\label{fig:try}
\end{figure}

\noindent
To see why, consider that effect capabilities are often scoped and therefore have a limited lifetime.
For instance a {\tt CanThrow[E]} capability would be
generated by a {\tt try} expression that catches {\tt E}. It is valid only as
long as the {\tt try} is executing. Figure~\ref{fig:try} shows an example of
capabilities for checked exceptions, both as source syntax on the left and with
compiler-generated implicit capability arguments on the right.
The following slight variation of this program would throw an unhandled exception since
the function {\tt f} is now evaluated only when the iterator's {\tt next} method is called,
which is after the {\tt\bf try} handling the exception has exited.

\begin{lstlisting}[language=scalaish]
    val it =
      try xs.iterator.map(f)
      catch case TooLarge => Iterator.empty
    it.next()
\end{lstlisting}

\noindent
A question answered in this paper is how to rule out {\tt Iterator}'s lazy {\tt map} statically while still allowing
{\tt List}'s strict {\tt map}. A large body of research exists that could address this problem
by restricting reference access patterns. Relevant techniques include
linear types \cite{DBLP:conf/ifip2/Wadler90},
rank 2 quantification \cite{launchbury-sabry:icfp97},
regions \cite{tofte1997region, grossman2002regions},
uniqueness types \cite{barendsenUniquenessTypingFunctional1996},
ownership types \cite{clarkeOwnershipTypesFlexible1998a,nobleFlexibleAliasProtection1998},
and second class values \cite{osvald2016gentrification}.
A possible issue with many of these approaches is their relatively high notational overhead, in particular
when dealing with polymorphism.

The approach we pursue here is different. Instead of restricting certain access patterns a priori, we focus
on describing what capabilities are possibly captured by values of a type.
At its core there are the following two interlinked concepts:

\newtheorem*{definition*}{Definition}

\begin{itemize}
  \item A {\em capturing type} is of the form $\{c_1, \ldots, c_n\}\;T$ where $T$ is a type
  and $\{c_1, \ldots, c_n\}$ is a {\em capture set} of capabilities.
  \item A {\em capability} is a parameter or local variable that has as
  type a capturing type with non-empty capture set. We call such
  capabilities {\em tracked} variables.
\end{itemize}

\noindent
Every capability gets its authority from some other, more sweeping capabilities which it captures.
The most sweeping capability, from which ultimately all others are derived, is ``$*$'', the {\em universal capability}.

As an example how capabilities are defined and used, consider a typical
{\em try-with-resources} pattern:

\begin{lstlisting}[language=scalaish]
  def usingFile[T](name: String, op: ({*} OutputStream) => T): T =
    val f = new FileOutputStream(name)
    val result = op(f)
    f.close()
    result

  val xs: List[Int] = ...
  def good = usingFile("out", f => xs.foreach(x => f.write(x)))
  def fail =
    val later = usingFile("out",
          f => (y: Int) => xs.foreach(x => f.write(x + y)))
    later(1)
\end{lstlisting}

\noindent
The {\tt usingFile} method runs a given operation {\tt op} on a freshly created
file, closes the file, and returns the operation's result. The method
enables an effect (writing to a file) and limits its validity (to until the file is closed).
Function {\tt good}
invokes {\tt usingFile} with an operation that writes each element of a given
list {\tt xs} to the file. By contrast, function {\tt fail} represents an illegal usage:
It invokes {\tt usingFile} with an operation that returns a function that, when invoked,
will write list elements to the file. The problem is that the writing happens
in the application {\tt later(1)} when the file has already been closed.

We can accept the first usage and reject the second by marking the
output stream passed to {\tt op} as a capability. This is done by prefixing its type
with {\tt \{*\}}. Our prototype implementation of capture checking will then
reject the second usage with an error message.

\noindent
This example used a capability parameter that was directly derived from
{\tt *}. But capabilities can also be derived from other non-universal
capabilities. For instance:
{\small
\begin{lstlisting}[language=scalaish]
  def usingLogFile[T](f: {*} OutputStream, op: ({f} Logger) => T): T =
    op(Logger(f))
\end{lstlisting}}
\noindent
The {\tt usingLogFile} method
takes an output stream (which is a capability) and an operation, which
gets passed a {\tt Logger}. The {\tt Logger} capability is
derived from the output stream capability, as can be seen from
its type {\tt \{f\} Logger}.

This paper develops a {\em capture calculus}, \CC, as a foundational type system that
allows reasoning about scoped capabilities. By sketching a prototype language design based
on this calculus, we argue that it is expressive enough to support a
wide range of usage patterns with very low notational overhead.
The paper makes the following specific contributions.

\begin{itemize}
\item We define a simple yet expressive type system for tracking captured capabilities in types. The calculus extends \fsub with capture sets of capabilities.
\item We prove progress and preservation of types relative to a small-step evaluation semantics. We also prove a capture prediction lemma that states that capture sets in types over-approximate
 captured variables in runtime values.
\item We report on a prototype language design in the form of an extension of Scala 3 and demonstrate
 its practicality in a number of examples that discuss both the embedding in fundamental datatypes and effect checking.
\end{itemize}
\noindent
The presented design is at the same time simple in theory and concise and
flexible in its practical application. We demonstrate that the following elements
are essential for achieving good usability:

\begin{itemize}
  \item Use reference-dependent typing, where a formal function parameter stands for the potential references captured by its argument \cite{oderskyEA2021safer, brachthaeuser2022effects}. This avoids the need to introduce separate binders for capabilities or effects. Technically, this means that references (but not general terms) can form part of types as
  members of capture sets. A similar approach is taken in the path-dependent typing discipline of DOT
  \cite{aminEssenceDependentObject2016,DBLP:conf/oopsla/RompfA16} and by reachability types for alias checking \cite{bao2021reachability}.
  \item Employ a subtyping discipline that mirrors subsetting of capabilities and that allows capabilities
   to be refined or abstracted. Subtyping of capturing types relies on a new notion of {\em subcapturing} that
   encompasses both subsetting (smaller capability sets are more specific than larger ones) and derivation
   (a capability singleton set is more specific than the capture set of the capability's type). Both dimensions
   are essential for a flexible modelling of capability domains.
\item Limit propagation of capabilities in instances of generic types where they cannot be accessed directly.
   This is achieved by boxing types when they enter a generic context and unboxing on every use site \cite{brachthaeuser2022effects}.
\end{itemize}
\noindent
The version of \CC presented here is similar to a system that was originally
proposed to make exception checking safe \cite{oderskyEA2021safer}.
Their system conjectured progress and preservation properties, but
did not have proofs. Compared to their system, the version in this paper
presents a fully worked-out meta theory with proofs of type soundness as
well as a semantic characterization of capabilities.
There are also some minor differences in the operational semantics, which
were necessary to make a progress theorem go through. Finally, we present
several use cases outside of exception handling, demonstrating the broad
applicability of the calculus.

Whereas many of our motivating examples describe applications in effect checking, the formal treatment
presented here does not mention effects. In fact, the effect domains are intentionally kept open
since they are orthogonal to the aim of the paper. Effects could be exceptions, file operations or
region allocations, but also algebraic effects, IO, or any sort of monadic effects. To express more
advanced control effects, one usually needs to add continuations to the operational semantics,
or use an implicit translation to the continuation monad. In short, capabilities can delimit
what effects can be performed at any point in the program but they by themselves don’t
perform an effect  \cite{marino09generic,gordon2020designing,liu2016,brachthaeuser2020effects,osvald2016gentrification}. For that, one needs a library or a runtime system that would be added
as an extension of \CC. Since \CC is intended to work with all such effect extensions,
we refrain from adding a specific extension to its operational semantics.

The rest of this paper is organized as follows.
Section~\ref{sec:informal} explains and motivates the core elements
of our calculus.
Section~\ref{sec:cc-calc} presents \CC. Section~\ref{sec:meta} lays out its meta-theory. Section~\ref{sec:examples}
illustrates the expressiveness of typing disciplines based on the calculus in examples. Section~\ref{sec:related} discusses
related work and Section~\ref{sec:conc} concludes.



%
%
%

\section{Informal Discussion}
\label{sec:informal}

This section motivates and discusses some of the key aspects of capture checking.
All examples are written in an experimental language extension of Scala 3 and
were compiled with our prototype implementation of a capture checker \cite{odersky2022cc-experiment}.

\subsection{Capability Hierarchy}

We have seen in the intro that every capability except $*$ is created from some
other capabilities which it retains in the capture set of its type. Here is an
example that demonstrates this principle:
\begin{lstlisting}[language=scalaish]
    class FileSystem

    class Logger(fs: {*} FileSystem):
      def log(s: String): Unit = ... // Write to a log file, using `fs`

    def test(fs: {*} FileSystem): {fs} LazyList[Int] =
      val l: {fs} Logger = new Logger(fs)
      l.log("hello world!")
      val xs: {l} LazyList[Int] =
        LazyList.from(1)
          .map { i =>
            l.log(s"computing elem # $i")
            i * i
          }
      xs
\end{lstlisting}
Here, the {\tt test} method takes a {\tt FileSystem} as a parameter. {\tt fs} is a capability since its type has a non-empty capture set.
The capability is passed to the {\tt Logger} constructor
and retained as a field in class {\tt Logger}. Hence, the local variable {\tt l} has type
{\tt \{fs\} Logger}: it is a {\tt Logger} which retains the {\tt fs} capability.

The second variable defined in {\tt test} is {\tt xs}, a lazy list that is obtained from
{\tt LazyList.from(1)} by logging and mapping consecutive numbers. Since the list is lazy,
it needs to retain the reference to the logger {\tt l} for its computations. Hence, the
type of the list is {\tt \{l\} LazyList[Int]}. On the other hand, since {\tt xs} only logs but does
not do other file operations, it retains the {\tt fs} capability only indirectly. That's why
{\tt fs} does not show up in the capture set of {\tt xs}.

Capturing types come with a subtype relation where types with ``smaller'' capture sets are subtypes of types with larger sets
(the {\em subcapturing} relation is defined in more detail below).
If a type {\tt T} does not have a capture set, it is called {\em pure}, and is a subtype of
any capturing type that adds a capture set to {\tt T}.

\subsection{Function Types}

The function type {\tt A => B} stands for a function that can capture arbitrary capabilities. We call such functions
{\em impure}. By contrast, the new single arrow function type {\tt A -> B} stands for a function that
cannot capture any capabilities, or otherwise said, is {\em pure}.
One can add a capture set in front of an otherwise pure function.
For instance, {\tt \{c, d\} A -> B} would be a function that can capture
capabilities {\tt c} and {\tt d}, but no others.

The impure function type {\tt A => B} is treated as an alias for {\tt \{*\} A -> B}.
That is, impure functions are functions that can capture anything.

Function types and captures both associate to the right, so
\lstinline[language=scalaish]!{c} A -> {d} B -> C!
is the same as
\lstinline[language=scalaish]!{c} (A -> {d} (B -> C))!.

Contrast with
\lstinline[language=scalaish]!({c} A) -> ({d} B) -> C!
which is a curried pure function over argument types that can capture {\tt c} and {\tt d}, respectively.

\paragraph{Note} The distinctions between pure vs impure function types do not apply to methods
(which are defined using \lstinline[language=scalaish]!def!). In fact, since methods are not values in Scala, they never
capture anything directly. References to capabilities in a method are instead counted
in the capture set of the enclosing object.

\subsection{Capture Checking of Closures}

If a closure refers to capabilities in its body, it captures these capabilities in its type. For instance, consider:
\begin{lstlisting}[language=scalaish]
      def test(fs: FileSystem): {fs} String -> Unit =
        (x: String) => Logger(fs).log(x)
\end{lstlisting}
Here, the body of {\tt test} is a lambda that refers to the capability {\tt fs}, which means that {\tt fs} is retained in the lambda.
Consequently, the type of the lambda is {\tt \{fs\} String -> Unit}.

\paragraph{Note} On the term level, function values are always written with {\tt =>} (or {\tt ?=>} for context functions). There is no syntactic
distinction for pure {\em vs} impure function values. The distinction is only made in their types.

\smallskip\noindent
A closure also captures all capabilities that are captured by the functions
it calls. For instance, in

\begin{lstlisting}[language=scalaish]
    def test(fs: FileSystem) =
      def f() = (x: String) => Logger(fs).log(x)
      def g() = f()
      g
\end{lstlisting}
the result of {\tt test} has type {\tt \{fs\} String -> Unit} even though function {\tt g} itself does not refer to {\tt fs}.

\subsection{Subtyping and Subcapturing}

Capturing influences subtyping. As usual we write $T_1 <: T_2$ to express that the type
$T_1$ is a subtype of the type $T_2$, or equivalently, that $T_1$ conforms to $T_2$. An
analogous {\em subcapturing} relation applies to capture sets. If $C_1$ and $C_2$ are capture sets,
we write $C_1 <: C_2$ to express that $C_1$ {\em is covered by} $C_2$, or, swapping the operands, that $C_2$ {\em covers} $C_1$.

Subtyping extends as follows to capturing types:

\begin{itemize}
  \item Pure types are subtypes of capturing types. That is, $T <: C\,T$, for any type $T$, capturing set $C$.
  \item For capturing types, smaller capture sets produce subtypes: $C_1\,T_1 <: C_2\,T_2$ if
   $C_1 <: C_2$ and $T_1 <: T_2$.
\end{itemize}

\noindent
A subcapturing relation $C_1 <: C_2$ holds if $C_2$ {\em accounts for} every element $c$ in $C_1$.
This means one of the following two conditions must be true:

\begin{itemize}
  \item $c \in C_2$,
  \item $c$'s type has capturing set $C$ and $C_2$ accounts for every element of $C$
     (that is, $C <: C_2$).
\end{itemize}

\noindent
{\bf Example.} Given

\begin{lstlisting}[language=scalaish]
      fs: {*} FileSystem
      ct: {*} CanThrow[Exception]
      l : {fs} Logger
\end{lstlisting}
we have
{\small
\begin{verbatim}
      {l}  <: {fs}     <: {*}
      {fs} <: {fs, ct} <: {*}
      {ct} <: {fs, ct} <: {*}
\end{verbatim}
}
\noindent
The set consisting of the root capability {\tt \{*\}} covers every other capture set. This is
a consequence of the fact that, ultimately, every capability is created from {\tt *}.

\subsection{Capture Tunneling}
Next, we discuss how type-polymorphism interacts with reasoning about capture.
To this end, consider the following simple definition of a {\tt Pair} class:
\begin{lstlisting}[language=scalaish]
      class Pair[+A, +B](x: A, y: B):
        def fst: A = x
        def snd: B = y
\end{lstlisting}
What happens if we pass arguments to the constructor of {\tt Pair} that capture capabilities?
\begin{lstlisting}[language=scalaish]
      def x: {ct} Int -> String
      def y: {fs} Logger
      def p = Pair(x, y)
\end{lstlisting}
Here the arguments {\tt x} and {\tt y} close over different capabilities {\tt ct} and {\tt fs},
which are assumed to be in scope.
So what should the type of {\tt p} be? Maybe surprisingly, it will be typed as:
\begin{lstlisting}[language=scalaish]
      def p: {} Pair[{ct} Int -> String, {fs} Logger] = Pair(x, y)
\end{lstlisting}
That is, the outer capture set is empty and does neither mention {\tt ct} nor {\tt fs},
even though the value {\tt Pair(x, y)} \emph{does} capture them.
So why don't they show up in its type at the outside?

While assigning {\tt p} the capture set {\tt \{ct, fs\}} would be sound, types would quickly grow inaccurate and unbearably verbose.
To remedy this, \CC{} performs \emph{capture tunneling}. Once a type variable is instantiated to a capturing type, the
capture is not propagated beyond this point. On the other hand, if the type variable is instantiated
again on access, the capture information ``pops out'' again.

Even though {\tt p} is technically untracked because its
capture set is empty, writing {\tt p.fst} would record a reference to the captured capability {\tt ct}.
So if this access was put in a closure, the capability would again form part of the outer capture set. \emph{E.g.},
\begin{lstlisting}[language=scalaish]
      () => p.fst : {ct} () -> {ct} Int -> String
\end{lstlisting}
In other words, references to capabilities ``tunnel through'' generic instantiations---from creation to access; they do not affect the capture set of the enclosing generic data constructor applications.
As mentioned above, this principle plays an important part in making capture checking concise and practical.
To illustrate, let us take a look at the following example:
\begin{lstlisting}[language=scalaish]
      def mapFirst[A,B,C](p: Pair[A,B], f: A => C): Pair[C,B] =
        Pair(f(p.x), p.y)
\end{lstlisting}
Relying on capture tunneling, neither the types of the parameters to {\tt mapFirst}, nor its result type need to be annotated with capture sets.
Intuitively, the capture sets do not matter for {\tt mapFirst}, since parametricity forbids it from inspecting the actual values inside the pairs.
If not for capture tunneling, we would need to annotate {\tt p} as {\tt \{*\} Pair[A,B]}, since both {\tt A} and {\tt B} and through them, {\tt p} can capture arbitrary capabilities.
In turn, this means that for the same reason, without tunneling we would also have {\tt \{*\} Pair[C,B]} as the result type.
This is of course unacceptably inaccurate.

Section~\ref{sec:cc-calc} describes the foundational theory on which capture checking is based.
It makes tunneling explicit through so-called {\em box} and
{\em unbox} operations. Boxing hides a capture set and unboxing recovers it.
The capture checker inserts virtual box and unbox operations based on actual and
expected types similar to the way the type checker inserts implicit conversions.
Boxing and unboxing has no runtime effect, so the insertion of these operations
is only simulated, but not kept in the generated code.

\comment{
Capture tunneling is in motivation analogous to effect tunneling \cite{zhang2016accepting,zhang2019abstraction} but
it applies to retained capabilities instead of executed effects and
it uses a different mechanism.
In a sense, capture tunneling works on a level above effect tunneling.
We get effect tunneling already from using capabilities. But where
effect tunneling requires capabilties to be second-class values,
we allow them to be returned in data.
Capture tunneling lets us omit capture annotations for generic portions of
such types and their constructors.\todo{Move to related?}}

\subsection{Escape Checking}

Following the principle of object capabilities, the universal capability {\tt *} should conceptually only be available as a parameter to the main program. Indeed, if it was available everywhere, capability checking would be undermined since one could mint new capabilities
at will. In line with this reasoning, some capture sets are restricted and must not contain the universal capability.

Specifically, if a capturing type is an instance of a type variable, that capturing type
is not allowed to carry the universal capability {\tt \{*\}}.\footnote{This follows since
type variables range over pure types, so {\tt *} must appear under a box. But rule \ruleref{Box} in
Figure~\ref{fig:cc-rules} restricts variables in boxed capture sets to be declared in the enclosing environment,
which does not hold for {\tt *}.}
There is a connection to tunneling here.
The capture set of a type has to be present in the environment when a type is instantiated from
a type variable. But {\tt *} is not itself available as a global entity in the environment. Hence,
this should result in an error.

Using this principle, we can show why the introductory example in Section~\ref{sec:introduction} reported an error. To recall, function {\tt usingFile} was declared like this:
\begin{lstlisting}[language=scalaish]
  def usingFile[T](name: String, op: ({*} FileOutputStream) => T): T = ...
\end{lstlisting}
The capture checker rejects the illegal definition of {\tt later}
{\small}\begin{lstlisting}[language=scalaish]
  val later = usingFile("out",
          f => (y: Int) => xs.foreach(x => f.write(x + y)))
\end{lstlisting}
with the following error message
\smallskip
{\small\begin{verbatim}
 |    val later = usingFile("out", f => (y: Int) => xs.foreach(x => f.write(x + y)))
 |                ^^^^^^^^^^^^^^^^^^^^^^^^^^^^^^^^^^^^^^^^^^^^^^^^^^^^^^^^^^^^^^^^^^
 |The expression's type {*} Int -> Unit is not allowed to capture the root capability `*`
 |This usually means that a capability persists longer than its allowed lifetime.
\end{verbatim}}
\smallskip
\noindent
This error message was produced by the following reasoning steps:

\begin{itemize}
 \item Parameter {\tt f} has type {\tt \{*\} FileOutputStream}, which makes it a capability.
 \item Therefore, the type of the expression
   \begin{lstlisting}[language=scalaish]
    (y: Int) => xs.foreach(x => f.write(x + y))
   \end{lstlisting}
   is {\tt \{f\} Int -> Unit}.
 \item Consequently, we assign the whole closure passed to {\tt usingFile} the
  dependent function type
   {\tt (f: \{*\} FileOutputStream) -> \{f\} () -> Unit}.
 \item The expected type of the closure is a simple, parametric, impure function type\newline
   {\tt (\{*\} FileOutputStream) => T},
   for some instantiation of the type variable {\tt T}.
 \item We cannot instantiate {\tt T} with {\tt\{f\} () -> Unit} since the expected
   function type is non-dependent. The smallest supertype that matches the expected type is thus\newline    {\tt (\{*\} FileOutputStream) => \{*\} Int -> Unit}.
 \item Hence, the type variable {\tt T} is instantiated to {\tt \{*\} () -> Unit},
   which is not allowed and causes the error.
\end{itemize}
\noindent

\subsection{Escape Checking of Mutable Variables}

Another way one could try to undermine capture checking would be to
assign a closure with a local capability to a global variable. For instance
like this\footnote{Mutable variables are not covered by the formal treatment of \CC. We
include the discussion anyway to show that escape checking can be generalized
to scope extrusions separate from result values.}:
\begin{lstlisting}[language=scalaish]
      var loophole: {*} () -> Unit = () => ()
      usingFile("tryEscape", f =>
        loophole = () => f.write(0)
      }
      loophole()
\end{lstlisting}
\noindent
We prevent such scope extrusions by imposing the restriction that mutable variables
cannot have types with universal capture sets.

One also needs to prevent returning or assigning a closure with a local capability in an argument of a parametric type. For instance, here is a
slightly more refined attack:
\begin{lstlisting}[language=scalaish]
      val sneaky = usingLogFile { f => Pair(() => f.write(0), 1) }
      sneaky.fst()
\end{lstlisting}
At the point where the {\tt Pair} is created, the capture set of the first argument is {\tt \{f\}}, which
is OK. But at the point of use, it is {\tt \{*\}}: since {\tt f} is no longer in scope
we need to widen the type to a supertype that does not mention it ({\em c.f.} the
explanation of avoidance in Section~\ref{sec:typing-rules}).
This causes an error, again, as the universal capability is not permitted to be
in the unboxed form of the return type ({\em c.f.} the precondition of \ruleref{Unbox} in
Figure~\ref{fig:cc-rules}).

\input{syntax.tex}

\section{The \CC calculus}
\label{sec:cc-calc}

The syntax of \CC is given in Figure~\ref{fig:syntax}. In short, it describes
a dependently typed variant of System \fsub in monadic normal form (MNF) with capturing types and boxes.

\paragraph{\bf Dependently typed:} Types may refer to term variables in their capture sets, which
introduces a simple form of (variable-)dependent typing. As a consequence, a function's
result type may now refer to the parameter in its capture set. To be able to express this,
the general form of a function type $\forall(x : U) T$ explicitly names the parameter $x$.
We retain the non-dependent syntax
$U \rightarrow T$ for function types as an abbreviation if the parameter is not mentioned in the result
type $T$.

Dependent typing is attractive since it means that we can refer to object capabilities directly
in types, instead of having to go through auxiliary region or effect variables. We thus avoid
clutter related to quantification of such auxiliary variables.

\paragraph{\bf Monadic normal form:} The term structure of \CC requires operands of applications to be
variables. This does not constitute a loss of expressiveness, since a general application
$t_1\,t_2$ can be expressed as $\Let{x_1}{t_1}{\Let{x_2}{t_2}{x_1\,x_2}}$.
This syntactic convention has advantages for variable-dependent typing. In particular, typing function application in such a calculus
requires substituting actual arguments for formal parameters. If arguments are
restricted to be variables, these substitutions are just variable/variable renamings,
which keep the general structure of a type. If arguments were arbitrary terms,
such a substitution would in general map a type to something that
was not syntactically a type. Monadic normal form \cite{hatcliff-danvy:popl94} is a slight generalization of
the better-known A-normal form (ANF) \cite{Sabry93reasoningabout} to allow arbitrary nesting
of let expressions. We use a here a variant of MNF where applications are over variables instead
of values.

A similar restriction to MNF was employed in DOT \cite{aminEssenceDependentObject2016},
the foundation of Scala's object model, for the same reasons. The restriction is
invisible to source programs, which can still be in direct style. For instance, the Scala
compiler selectively translates a source expression in direct style to MNF if a non-variable
argument is passed to a dependent function. Type checking then takes place on the translated version.

\paragraph{\bf Capturing Types:} The types in \CC are stratified as {\em shape types} $S$
and regular types $T$. Regular types can be shape types or capturing types $\Capt{\{x_1, \ldots, x_n\}}{S}$.
Shape types are made up from the usual type constructors in \fsub plus boxes. We freely use
shape types in place of types, assuming the equivalence $\{\}\,S \equiv S$.

\paragraph{\bf Boxes:}
Type variables $X$ can be bounded or instantiated only with shape types, not with regular types.
To make up for this restriction, a regular type $T$ can be encapsulated in a shape type by
prefixing it with a box operator $\Boxed{T}$. On the term level,
$\Box~x$ injects a variable into a boxed type.
A variable of boxed type is unboxed using the syntax $\Unboxed{C}{x}$ where $C$
is the capture set of the underlying type of $x$. We have seen in Section~\ref{sec:informal}
that boxing and unboxing allow a kind of capability tunneling
by omitting capabilities when values of parametric types are constructed and charging these capabilities instead
at use sites.

\paragraph{\bf System \fsub:} We base \CC on a standard type system that supports
the two principal forms of polymorphism, subtyping and universal.

Subtyping
comes naturally with capabilities in capture sets. First, a type capturing
fewer capabilities is naturally a subtype of a type capturing more capabilities,
and pure types are naturally subtypes of capturing types. Second, if capability
$x$ is derived from capability $y$, then a type capturing $x$ can be seen
as a subtype of the same type but capturing $y$.

Universal polymorphism poses
specific challenges when capture sets are introduced which are addressed in \CC
by the stratification into shape types and regular types and the box/unbox
operations that map between them.

Note that the only form of term dependencies in \CC relate to capture sets in types. If
we omit capture sets and boxes, the calculus is equivalent to standard \fsub, despite
the different syntax. We highlight in the figures the essential additions wrt \fsub with a grey background.

\CC is intentionally meant to be a small and canonical core calculus that does not cover
higher-level features such as records, modules, objects, or classes. While these features
are certainly important, their specific details are also somewhat more varied
and arbitrary than the core that's covered. Many different systems
can be built on \CC, extending it with various constructs to organize code and
data on higher-levels.

\paragraph{\bf Capture Sets}
Capture sets $C$ are finite sets of variables of the form $\cset{x_1, \hdots,
x_n}$. We understand $\rcap$ to be a special variable that can appear in
capture sets, but cannot be bound in $\Gamma$. We write $C \setminus x$
as a shorthand for $C \setminus \cset{x}$.

Capture sets of closures are determined using a function $\cv$ over terms.

\begin{definition*}[Captured Variables]
{\em The captured variables $\cv(t)$ of a term $t$ are given as follows.}
\[\begin{array}{lcl}
    \cv(\lambda(x : T)t) &=& \cv(t) \backslash x \\
    \cv(\lambda[X \sub S]t) &=& \cv(t) \\
    \cv(x) &=& \{x\} \\
    \cv(\Let{x}{v}{t}) &=& \cv(t) \gap\gap\gap \kw{if} x \notin \cv(t) \\
    \cv(\Let{x}{s}{t}) &=& \cv(s) \cup \cv(t)  \backslash x\\
    \cv(x\,y) &=& \{x, y\} \\
    \cv(x[S]) &=& \{x\} \\
    \cv(\Boxed{x}) &=& \{\} \\
    \cv(\Unboxed{C}{x}) &=& C \cup \{x\}
\end{array}\]
\end{definition*}

\noindent
The definitions of captured and free variables of a term are very similar, with the following three differences:

\begin{enumerate}
\item Boxing a term $\Boxed{x}$ obscures $x$ as a captured variable.
\item Dually, unboxing a term $\Unboxed{C}{x}$ counts the variables in $C$ as
  captured.
\item In an evaluated let binding $\Let{x}{v}{t}$, the captured variables of $v$
  are counted only if $x$ is a captured variable of $t$.
\end{enumerate}
The first two rules encapsulate the essence of box-unbox pairs:
Boxing a term obscures its captured variable and makes it necessary
to unbox the term before its value can be accessed;
unboxing a term presents variables that were obscured when boxing.


\input{cc-rules.tex}

\medskip\noindent
Figure~\ref{fig:cc-rules} presents typing and evaluation rules for \CC. There are
four main sections on subcapturing, subtyping, typing, and evaluation. These
are explained in the following.

\subsection{Subcapturing}

Subcapturing establishes a preorder relation on capture sets that gets propagated
to types. Smaller capture sets with respect to subcapturing lead to smaller types
with respect to subtyping.

The relation is defined by three rules.
The first two rules \ruleref{sc-set} and \ruleref{sc-elem} establish that subsets imply subcaptures.
That is, smaller capture sets subcapture larger ones. The last rule \ruleref{sc-var}
is the most interesting since it reflects an essential property of object capabilities.
It states that a variable $x$ of capturing type $C\;S$ generates a capture set $\{x\}$ that
subcaptures the capabilities $C$ with which the variable was declared. In a sense,
\ruleref{sc-var} states a monotonicity property: a capability
refines the capabilities from which it is created. In particular, capabilities
cannot be created from nothing. Every capability needs to be derived from some
more sweeping capabilities which it captures in its type.

The rule also validates our definition of capabilities as variables with non-empty capture
sets in their types. Indeed, if a variable is defined as $x: \{\}\;S$, then by
\ruleref{sc-var} we have $\{x\} <: \{\}$. This means that the variable can be disregarded
in the formation of $\cv$, for instance. Even if $x$ occurs in a term, a capture set
with $x$ in it is equivalent (with respect to mutual subcapturing) to a capture set without.
Hence, $x$ can safely be dropped without affecting subtyping or typing.

Rules \ruleref{sc-set} and \ruleref{sc-elem} mean that if set $C$ is a subset of
$C'$, we also have $C \sub C'$. But the reverse is not true.
For instance, with \ruleref{sc-var} we can derive the following relationship assuming
lambda-bound variables $x$ and $y$:
$$
  x:\CS{\UC}\Top, y:\CAPT{x}\Top \ts \cset{y} \sub \cset{x}
$$

\noindent
Intuitively this makes sense, as $y$ can capture no more than $x$.
However, we {\em cannot} derive $\cset{x} \sub \cset{y}$, since arguments passed
for $y$ may in fact capture {\it less} than $x$, e.g. they could be pure.

While there are no subcapturing rules for top or bottom capture sets, we can still establish:

\begin{proposition} If $C$ is well-formed in $\Gamma$, then $\Gamma |- \{\} <: C <: \UC$.
\end{proposition}

A proof is enclosed in the appendix.

\begin{proposition} The subcapturing relation $\Gamma \ts \_ <: \_$ is a preorder.
\end{proposition}

\begin{proof} We can show that transitivity and reflexivity are admissible.
\end{proof}

\subsection{Subtyping}

The subtyping rules of \CC are very similar to those of \sysfsub, with
the only significant addition being the rules for capturing and boxed types.
Note that as $S \equiv \{\}\;S$, both transitivity and reflexivity apply to shape
types as well.
\ruleref{capt} allows comparing types that have capture sets, where smaller
capture sets lead to smaller types. \ruleref{boxed} propagates subtyping relations
between types to their boxed versions.

\subsection{Typing}\label{sec:typing-rules}

Typing rules are again close to \sysfsub, with differences to account for
capture sets.

Rule \ruleref{var} is the basis for capability refinements. If $x$ is declared with type $C\;S$, then
the type of $x$ has $\{x\}$ as its capture set instead of $C$.
The capturing set $\{x\}$ is more specific
than $C$, in the subcapturing sense. Therefore, we can recover the capture set $C$ through subsumption.

Rules \ruleref{abs} and \ruleref{tabs} augment the abstraction's type with a
capture set that contains the captured variables of the term. Through subsumption and rule
\ruleref{sc-var}, untracked variables can immediately be removed from this set.

The \ruleref{app} rule substitutes references to the function parameter with the
argument to the function. This is possible since arguments are guaranteed to be variables.
The function's capture set $C$ is disregarded, reflecting the principle that
the function closure is consumed by the application.
Rule \ruleref{tapp} is analogous.

{\em Aside:}
A more conventional version of \ruleref{tapp} would be

\infrule[\ruledef{tapp'}]{%
  \Gamma \ts x \typ \Capt{C}{\forall[X \sub S']T} \gap
  \Gamma \ts S \sub S'
}{%
  \Gamma \ts x[S] \typ [X := S]T}

That formulation is equivalent to \ruleref{tapp} in the sense that either rule
is derivable from the other, using subsumption and contravariance of type
bounds.

Rules \ruleref{box} and \ruleref{unbox} map between boxed and unboxed types.
They require all members of the capture set under the box to be bound in the environment
$\Gamma$. Consequently, while one can create a boxed type with $\UC$ as its capture set
through subsumption, one cannot unbox values of this type.
This property is fundamental for ensuring scoping of capabilities.

\paragraph{\bf Avoidance}
As is usual in dependent type systems, Rule \ruleref{let} has as a side condition that the bound variable $x$ does not appear
free in the result type $U$. This so called {\em avoidance} property is usually attained
through subsumption. For instance consider an enclosing capability $c: T_1$ and the term
\[
  \Let x {\lambda(y: T_2).c} \lambda(z: \{x\}\,T_3).z
\]
The most specific type of $x$ is $\{c\}\,(T_2 \rightarrow T_1)$ and the most specific type
of the body of the let is $\forall(z: \{x\}\,T_3).\{z\}\,T_3$. We need to find a supertype
of the latter type that does not mention $x$. It turns out the most specific such type
is $T_3 \rightarrow \{c\} T_3$, so that is a possible type of the let, and it should
be the inferred type.

In general there is always a most specific avoiding type for a \ruleref{let}:

\begin{proposition}\label{lemma:avoiding-let-types-exist}
Consider a term $\Let x s t$ in an environment $\Gamma$ such that $\Gamma \ts s : T_1$
and $\Gamma, x: T_1 \ts t: T_2$. Then there exists a minimal (wrt $<:$) type
$T_3$ such that $T_2 <: T_3$ and $x \notin \fv(T_3)$
\end{proposition}

\begin{proof}
  Without loss of generality, assume that $\Gamma \ts s : \Capt{C_s}{U}$
  is the most specific typing for $s$ in $\Gamma$
  and $\Gamma, x : \Capt{C_s}{U} \ts t : T$ is the most specific
  typing for $t$ in the context of the body of the let, namely $\Gamma, x : \Capt {C_s}{U}$.  Let $T'$
  be constructed from $T$ by replacing $x$ with $C_s$ in covariant capture set positions and by replacing $x$
  with the empty set in contravariant capture set positions.  Then for every type $V$ avoiding $x$
  such that $\Gamma, x : \Capt{C_s}{S} \ts T \sub V$, $\Gamma \ts T' \sub V$.
  This is shown by a straightforward structual induction, which we give in the
  appendix in Section \ref{lemma:avoiding-let-types-exist:proof}.
\end{proof}

\subsection{Well-formedness}

Well-formedness $\Gamma \ts T \wf$ is equivalent to well-formedness in \sysfsub in that free variables
in types and terms must be defined in the environment, except that capturing types
may mention the universal capability $\rcap$ in their capture sets:

\infrule[\ruledef{capt-wf}]{%
  \Gamma \ts S \wf \gap C \subseteq dom(\Gamma) \cup \{\rcap\}
}{%
  \Gamma \ts \Capt{C}{S} \wf}

\subsection{Evaluation}

Evaluation is defined by a small-step reduction relation. This relation is quite different
from usual reduction via term substitution. Substituting values for variables would break the
monadic normal form of a program. Instead, we reduce the right hand sides of let-bound variables in
place and lookup the bindings in the environment of a redex.

Every redex is embedded in an outer {\em store context} and an inner {\em evaluation context}.
These represent orthogonal decompositions of let bindings. An evaluation context $e$ always
puts the focus $[]$ on the right-hand side $t_1$ of a let binding $\Let x {t_1} {t_2}$. By contrast,
a store context $\sigma$ puts the focus on the following term $t_2$ and requires that $t_1$ is evaluated.

The first three rules --- \ruleref{apply}, \ruleref{tapply}, \ruleref{open} --- rewrite
simple redexes: applications, type applications and unboxings. Each of these rules
looks up a variable in the enclosing store and proceeds based on the value that was found.

The last two rules are administrative in nature. They both deal with evaluated
{\bf let}s in redex position. If the right hand side of the {\bf let} is a variable, the {\bf let}
gets expanded out by renaming the bound variable using \ruleref{rename}. If it is a value, the {\bf let} gets
lifted out into the store context using \ruleref{lift}.

\begin{proposition}[] Evaluation is deterministic. For any term $t_1$ there is at most one
  term $t_2$ such that $t_1 \reduces t_2$.
\end{proposition}
\proof
  By a straightforward inspection of the reduction rules and definitions of contexts.

\section{Metatheory}\label{sec:meta}

We prove that \CC is sound through the standard progress and preservation
theorems.
The proofs for all the lemmas and theorems stated in this section are provided
in the appendix.

In order to prove both progress and preservation, we need technical lemmas that
allow manipulation of typing judgments for terms under store and evaluation
contexts. To state these lemmas, we first need to define what it means for
typing and store contexts to match, which we do in Figure \ref{fig:matching-env}.

\begin{figure}[H]
  \begin{minipage}{0.5\textwidth}
    \infrule{%
      \Gamma,x:T \ts \sigma \sim \Delta \\
      \Gamma \ts v \typ T \qquad x \not\in \fv(T)
    }{%
      \Gamma \ts \Let{x}{v}{\sigma} \sim x:T,\Delta
    }
  \end{minipage}%
  \begin{minipage}{0.5\textwidth}
    \infax{%
      \Gamma \ts [\,] \sim \cdot
    }
  \end{minipage}
  \caption{Matching environment {\fbox{$\Gamma \ts \sigma \sim \Delta$}}}
  \label{fig:matching-env}
\end{figure}

\noindent
Having $\Gamma \ts \sigma \sim \Delta$ lets us know that $\sigma$ is well-typed in $\Gamma$ if we use $\Delta$ as the types of the bindings.
Using this definition, we can state the following four lemmas, which also illustrate how the store and evaluation contexts interact with typing:

\providecommand\forcenl{\,\\}

\begin{definition}[Evaluation context typing ($\Gamma |- e : U => T$)]
  We say that $\Gamma$ types an evaluation context $e$ as $U => T$ if $\Gamma,x:U |- \evalof{x} : T$.
\end{definition}

\begin{lemma}[Evaluation context typing inversion]\forcenl
  \indent
  $\Gamma \ts \evalof{s} : T$ implies that for some $U$ we have $\Gamma |- e : U => T$ and $\Gamma |- s : U$.
\end{lemma}

\begin{lemma}[Evaluation context reification]\forcenl
  \indent
  If both $\Gamma \ts e : U => T$ and $\Gamma \ts s : U$, then $\Gamma \ts \evalof{s} : T$.
\end{lemma}

\begin{lemma}[Store context typing inversion]\forcenl
  \indent
  $\Gamma \ts \storeof{t} \typ T$ implies that for some $\Delta$ we have $\Gamma \ts \sigma \sim \Delta$ and $\Gamma,\Delta \ts t \typ T$.
\end{lemma}

\begin{lemma}[Store context reification]\forcenl
  \indent
  If both $\Gamma,\Delta \ts t \typ T$ and $\Gamma \ts \sigma \sim \Delta$, then also $\Gamma \ts \storeof{t} \typ T$.
\end{lemma}

We can now proceed to our main theorems; their statements differ slightly from \sysfsub,
as we need to account for our monadic normal form.
Our preservation theorem captures that the important type to preserve is the one assigned to the term under the store.
It is stated as follows:

\begin{theorem}[Preservation]
  If we have $\Gamma \ts \sigma \sim \Delta$ and $\Gamma,\Delta \ts t : T$,
  then $\storeof{t} \reduces \storeof{t'}$
  implies that $\Gamma,\Delta \ts t' : T$.
\end{theorem}

To neatly state the progress theorem, we first need to define canonical store-plug
splits, which simply formally express that the store context wrapping a term is as
large as possible:

\begin{definition}[Canonical store-plug split]
  We say that a term of the form $\storeof{t}$ is a {\it canonical split} (of
  the entire term into store context $\sigma$ and the plug $t$) if $t$ is not of
  the form $\Let{x: T}{v}{t'}$.
\end{definition}

\noindent With this definition, we can now state the progress theorem:

\begin{theorem}[Progress]
  If $\ts \storeof{\evalof{t}} \typ T$
  and $\storeof{\evalof{t}}$ is a canonical store-plug split,
  then either $\evalof{t} = a$,
  or there exists $\storeof{t'}$ such that $\storeof{\evalof{t}} \reduces \storeof{t'}$.
\end{theorem}

\noindent The lemmas needed to prove progress and preservation are for the most
part standard. As our calculus is term-dependent, we also need to account for
term substitution affecting both environments and types, not only terms. For
instance, the lemma stating that term substitution preserves typing is expressed
as follows:

\begin{lemma}[Term substitution preserves typing]\forcenl
  \indent If $\Gamma,x:U,\Delta \ts t \typ T$
  and $\Gamma \ts y \typ U$,
  then $\Gamma,[x \mapsto y]\Delta \ts [x \mapsto y]t \typ [x \mapsto y]T$.
\end{lemma}

In this statement, we can also see that we only consider substituting one term
variable for another, due to MNF. Using MNF affects other parts of the proof as
well -- in addition to typical canonical forms lemmas, we also need to show that
looking up the value bound to a variable in a store preserves the types we can
assign to the variable:

\begin{lemma}[Variable lookup inversion]\forcenl
  \indent
  If we have both $\Gamma \ts \sigma \sim \Delta$ and $x : C\,R \in \Gamma,\Delta$,
  then $\sigma(x) = v$ implies that $\Gamma,\Delta \ts v \typ \CS{C}R$.
\end{lemma}

\subsubsection*{Capture sets and captured variables}

Our typing rules use $\cv$ to calculate the capture set that should be assigned
to terms. With that in mind, we can ask the question: what is the exact
relationship between captured variables and capture sets we use to type the terms?

Because of subcapturing, this relationship is not as obvious as it might seem.
For fully evaluated terms (of the form $\sigma[a]$), their captured variables
are the most precise capture set they can be assigned. The following lemma
states this formally:

\begin{lemma}[Capture prediction for answers]\label{lemma:capt-predict-a}
  If $\Gamma \ts \sigma[a] \typ \CS{C}S$, then $\Gamma \ts \cv(\sigma[a]) \sub C$.
\end{lemma}



If we start with an unreduced term $\sigma[t]$, then the situation becomes more
complex. It can mention and use capabilities that will not be reflected in the
capture set at all -- for instance, if $t = x\,y$, the capture set of $x$ is
irrelevant to the type assigned to $t$ by \ruleref{app}.
However, if $\sigma[t]$ reduces fully to a term of the form
$\sigma[\sigma'[a]]$, the captured variables of
$\sigma'[a]$ will correspond to capture sets we could assign to $t$.

In other words, the capture sets we assign to unevaluated terms under a store
context predict variables that will be captured by the answer those terms reduce
to. Formally we can express this as follows:

\begin{lemma}[Capture prediction for terms]\label{lemma:capt-predict-t}\forcenl
  \indent Let $\ts \sigma \sim \Delta$ and $\Delta \ts t \typ \CS{C}R$.
  Then $\sigma[t] -->^{*} \sigma[\sigma'[a]]$
  implies that $\Delta \ts \cv(\sigma'[a]) \sub C$.
\end{lemma}

\subsection{Correctness of boxing}
Boxing capabilities allows temporarily hiding them from capture sets.
Given that we can do so, one may be inclined to ask: what is the correctness criterion for our typing rules for box and unbox forms?

Intuitively, we should not be able to ``smuggle in'' a capability by using boxes: all the capabilities of a term should be somehow accounted for.
A basic criterion our typing rules should satisfy is that a term that boxes and immediately unboxes a capability should have a cv at least as wide than that of the capability,
i.e.\ a cv that accounts for the capability. Formally we can state this property as follows:

\begin{proposition}
  Let $|- \sigma \ts \Delta$ and let $t = \kw{let} y = \Boxed{x} \kw{in} \Unboxed{C}{y}$
  such that for some $e,T$, we have $\Delta |- \storeof{\evalof{t}} : T$ . Then we also have:
  $$
    \Delta |- \{x\} <: C = \cv(t)
  $$
\end{proposition}

It is straightforward to verify that this holds based on the progress theorem and by induction on the typing derivation.
Ideally, our correctness criterion would speak about arbitrary terms, by characterising their captured variables alongside reduction paths.
Since this clearly doesn't make sense for closed terms, we will need one additional definition: a way of introducing well-behaved capabilities.
A \emph{platform} $\Psi$ is an outer portion of the store which formally represents primitive capabilities. It is defined as follows:
$$\begin{array}{lcl}
  \Psi &::=& \Let x v {\Psi}    \gap\mbox{if}\ \fv(v) = \{\} \\
      &|& [\,]
\end{array}$$

\noindent
Since a platform $\Psi$ is meant to introduce primitive capabilities, it doesn't make sense to type them with any other capture set than $\UC$.
The following definition captures this idea:

\begin{definition*}[Well-typed program]
  We say that a term $\Psi[t]$ is a well-typed program if for some typing context $\Delta$ such that $|- \Psi \sim \Delta$
  we have $\Delta |- t : T$, and for all $x \in \dom \Psi$ we have $x : \UC\,S \in \Delta$.
\end{definition*}

The capabilities introduced by the platform of a well-typed program are well-behaved in the sense that their values cannot reference any other variable, and their capture sets are $\UC$.
Given $\cv$ and platform, we can state our desired lemma as follows:

\begin{lemma}[Program authority preservation]
  Let $\Psi[t]$ be a well-typed program such that $\Psi[t] --> \Psi[t']$.
  Then $\cv(t')$ is a subset of $\cv(t)$.
\end{lemma}

The name of this lemma hints at the property we want to demonstrate next, which offers another perspective at what it means for boxes to be correctly typed.
This property will capture the connection between the $\cv$ of a term and the capabilities it may use during evaluation.
We make intuitive usage of this connection when we say, for instance, that a function $f$ typed as $\{\} \Unit -> \Unit$ cannot perform any effects when called.
Formally, it cannot do so because when we reduce a term of the form $f\,x$, based on \ruleref{abs} we know it will be replaced by some term such that $\cv(t) = \{\}$.
Here is where we make an intuitive jump: since $\cv(t)$ is empty and we do not pass anything to $t$, it should not use any capabilities during evaluation.

\newcommand\used{\ensuremath{\mathrm{used}}}

To state this formally, first we need to define which capabilities are used during reduction.
All formal capabilities are abstractions, so a natural definition is to say that a capability $x$ was used if it was applied, i.e. if a term of the form $x\,y$ was reduced.
To speak about this, we will define $\used(t --> s)$ as function from reduction derivations to sets of variables, as follows:
$$\begin{array}{lll}
  \used( t_1 --> t_2 --> \cdots --> t_n ) &=& \used(t_1 --> t_2) \cup \used(t_2 --> \cdots --> t_n) \\
  \used( \sigma[e[x\,y]] --> \sigma[t] ) &=& \{x\} \\
  \used( \sigma[e[x\,[T]]] --> \sigma[t] ) &=& \{x\} \\
  \used(t_1 --> t_2) &=& \{\}   \gap\gap\mbox{(otherwise)}
\end{array}$$
The last case applies to rules (OPEN), (RENAME), (LIFT).

Now, let us consider how exactly the $\cv$ of a term is connected to the capabilities it will use during evaluation.
If we disregard boxes, then clearly a term $t$ can only use a capability $x$ if it is free in $t$, which also means that it is in $\cv(t)$.
However, boxing a capability hides it from the $\cv$.
If we disregard unbox forms, this is again not a problem: a capability under a box cannot be used unless the box is opened first.
Now, if we consider unbox forms $\Unboxed C x$, we encounter the essence of the problem: $\cv(\Unboxed C x) = C$.
On the term level, we have no guarantee that the capture sets of $C$ and $x$ are connected.
The \ruleref{box} and \ruleref{unbox} typing rules ensure that if a term gains access to a capability by opening a box,
then the ``key'' used to open the box (the $C$ in $\Unboxed C x$) must account for the boxed capability.
Thanks to them, we can say that $\cv(t)$ is the \emph{authority} of the term $t$: during reduction, $t$ can only use capabilities contained in $\cv(t)$.
We formally state this property as follows:

\begin{lemma}[Used capability prediction]
  Let $\Psi[t]$ be a well-typed program that reduces in zero or more steps to $\Psi[t']$.
  Then the primitive capabilities used during the reduction are a subset of the authority of $t$.
\end{lemma}

\section{Examples}\label{sec:examples}

We have implemented a type checker for \CC as an extension of the Scala 3 compiler
to enable experimentation with larger code examples. Notably, our extension
infers which types must be boxed, and automatically generates boxing and
unboxing operations when values are passed to and returned from instantiated
generic datatypes, so none of these technical details appear in the actual user-written
Scala code. In this section, we present examples that demonstrate the usability of the language.
In Section~\ref{sec:eg:list}, we remain close to the core calculus by
encoding lists using only functions; here, we still show the boxed types
and boxing and unboxing operations that the compiler infers in gray,
though they are not in the source code. In Sections~\ref{sec:eg:stackalloc} and~\ref{sec:eg:collections},
we use additional Scala features in larger examples,
to implement stack allocation and polymorphic data structures.
For these examples, we present the source code without cluttering it
with the boxing operations inferred by the compiler.

\subsection{Church-Encoded Lists}
\label{sec:eg:list}

Using the Scala prototype implementation of \CC, the B{\"o}hm-Berarducci encoding~\cite{bohmAutomaticSynthesisTyped1985}
of a linked list data structure can be implemented and typed as follows.  Here, a list
is represented by its right fold function:

\newcommand{\infer}[1]{\textcolor{gray}{#1}}
\newcommand{\topbnd}{\infer{<: $\Box$ \cset{*} Any}}
\newcommand{\infbox}{\infer{$\Box$}}
\newcommand{\infunbox}[1]{\infer{\cset{#1} $\circ\!-$}}

\begin{lstlisting}[language=scalaish, escapeinside={(*@}{@*)}]
    type Op[T (*@\topbnd@*), C (*@\topbnd @*)] =
      (v: T) => (s: C) => C

    type List[T (*@\topbnd@*)] =
      [C (*@\topbnd@*)] -> (op: Op[T, C]) -> {op} (s: C) -> C

    def nil[T (*@\topbnd@*)]: List[T] =
      [C (*@\topbnd@*)] => (op: Op[T, C]) => (s: C) => s

    def cons[T (*@\topbnd@*)](hd: T, tl: List[T]): List[T] =
      [C (*@\topbnd@*)] => (op: Op[T, C]) => (s: C) => op(hd)(tl(*@\infer{[C]}@*)(op)(s))
\end{lstlisting}


\noindent
A list inherently captures any capabilities that
may be captured by its elements. Therefore, naively,
one may expect the capture set of the list to include the
capture set of the type \lstinline{T} of its elements.
However, boxing and unboxing
enables us to elide the capture
set of the elements from the capture  set of the containing list.
When constructing a list using {\tt cons}, the elements must be boxed:
\begin{lstlisting}[language=scalaish, escapeinside={(*@}{@*)}]
    cons((*@\infbox@*) 1, cons((*@\infbox@*) 2, cons((*@\infbox@*) 3, nil)))
\end{lstlisting}

\noindent
A \lstinline{map} function over the list can be implemented and typed
as follows:
\begin{lstlisting}[language=scalaish, escapeinside={(*@}{@*)}]
  def map[A (*@\topbnd@*), B (*@\topbnd@*)](xs: List[A])(f: {*} A -> B): List[B] =
    (*@\infunbox{}@*) xs[(*@\infbox@*) List[B]]((hd: A) => (tl: List[B]) => cons(f(hd), tl))(nil)
\end{lstlisting}

\noindent
The mapped function \lstinline{f} may capture any capabilities,
as documented by the capture set \lstinline+{*}+ in its type.
However, this does not affect the type of \lstinline{map} or its
result type \lstinline{List[B]}, since the mapping is strict,
so the resulting list does not capture any capabilities captured
by \lstinline{f}. If a value returned by the function \lstinline{f}
were to capture capabilities, this would be reflected in its type,
the concrete type substituted for the type variable \lstinline{B},
and would therefore be reflected in the concrete instantiation of the
result type \lstinline{List[B]} of \lstinline{map}.

%

\subsection{Stack Allocation\label{sec:eg:stackalloc}}
Automatic memory management using a garbage collector is convenient and prevents
many errors, but it can impose significant performance overheads
in programs that need to allocate large numbers of short-lived objects.
If we can bound the lifetimes of some objects to coincide with
a static scope, it is much cheaper to allocate those objects on a stack as follows:
\footnote{For simplicity, this example is neither thread nor exception safe.}
\begin{lstlisting}[language=scalaish]
    class Pooled

    val stack = mutable.ArrayBuffer[Pooled]()
    var nextFree = 0

    def withFreshPooled[T](op: Pooled => T): T =
      if nextFree >= stack.size then stack.append(new Pooled)
      val pooled = stack(nextFree)
      nextFree = nextFree + 1
      val ret = op(pooled)
      nextFree = nextFree - 1
      ret
\end{lstlisting}

The {\tt withFreshPooled} method calls the provided function {\tt op}
with a freshly stack-allocated instance of class {\tt Pooled}.
It works as follows.
The {\tt stack} maintains a pool of already allocated instances of
{\tt Pooled}. The {\tt nextFree} variable records the offset of the first
element of {\tt stack} that is available to reuse; elements before it are
in use. The {\tt withFreshPooled} method first checks whether
the {\tt stack} has any available instances; if not, it adds one to the stack.
Then it increments {\tt nextFree} to mark the first available instance as used,
calls {\tt op} with the instance, and decrements {\tt nextFree} to mark the
instance as freed. In the fast path, allocating and freeing an instance of
{\tt Pooled} is reduced to just incrementing and decrementing
the integer {\tt nextFree}.

However, this mechanism fails if the instance of {\tt Pooled} outlives
the execution of {\tt op}, if {\tt op} captures it in its result. Then
the captured instance may still be accessed while at the same time also
being reused by later executions of {\tt op}. For example, the following
invocation of {\tt withFreshPooled} returns a closure that accesses the
{\tt Pooled} instance when it is invoked on the second line, after the
{\tt Pooled} instance has been freed:
\begin{lstlisting}[language=scalaish]
    val pooledClosure = withFreshPooled(pooled => () => pooled.toString)
    pooledClosure()
\end{lstlisting}
\noindent
Using capture sets, we can prevent such captures and ensure the safety of stack allocation just by marking the {\tt Pooled} instance as tracked:
\begin{lstlisting}[language=scalaish]
    def withFreshPooled[T](op: (pooled: {*} Pooled) => T): T =
\end{lstlisting}
Now the {\tt pooled} instance can be captured only in values
whose capture set accounts for {\tt \{pooled\}}. The type
variable {\tt T} cannot be instantiated with such a capture set
because {\tt pooled} is not in scope outside of {\tt withFreshPooled},
so only {\tt \{*\}} would account for {\tt \{pooled\}}, but we disallowed
instantiating a type variable with {\tt \{*\}}.
With this declaration of {\tt withFreshPooled}, the above {\tt pooledClosure}
example is correctly rejected, while the following safe example is allowed:
\begin{lstlisting}[language=scalaish]
    withFreshPooled(pooled => pooled.toString)
\end{lstlisting}

\subsection{Collections}
\label{sec:eg:collections}

In the following examples we show that a typing discipline based on \CC can be lightweight enough
to make capture checking of operations on standard collection types practical. This is important,
since such operations are the backbone of many programs. All examples compile with our current
capture checking prototype \cite{odersky2022cc-experiment}.

We contrast the APIs of common operations on Scala's standard collection types {\tt List} and
{\tt Iterator} when capture sets are taken into account.
Both APIs are expressed as Scala 3 extension methods \cite{dotty-extension-methods}
over their first parameter. Here is the {\tt List} API:

\begin{lstlisting}[language=scalaish]
    extension [A](xs: List[A])
      def apply(n: Int): A
      def foldLeft[B](z: B)(op: (B, A) => B): B
      def foldRight[B](z: B)(op: (A, B) => B): B
      def foreach(f: A => Unit): Unit
      def iterator: Iterator[A]
      def drop(n: Int): List[A]
      def map[B](f: A => B): List[B]
      def flatMap[B](f: A => `{*}` IterableOnce[B]): List[B]
      def ++[B >: A](xs: `{*}` IterableOnce[B]): List[B]
\end{lstlisting}
\noindent
Notably, these methods have almost exactly the same signatures as their versions
in the standard Scala collections library. The only differences concern the
arguments to {\tt flatMap}  and {\tt ++} which now admit an {\tt IterableOnce}
argument with an arbitrary capture set. Of course, we could have left out the
{\tt \{*\}} but this would have needlessly restricted the argument to
non-capturing collections.

Contrast this with some of the same methods for iterators:

\begin{lstlisting}[language=scalaish]
  extension [A](it: Iterator[A])
    def apply(n: Int): A
    def foldLeft[B](z: B)(op: (B, A) => B): B
    def foldRight[B](z: B)(op: (A, B) => B): B
    def foreach(f: A => Unit): Unit

    def drop(n: Int): `{it}` Iterator[A]
    def map[B](f: A => B): `{it, f}` Iterator[B]
    def flatMap[B](f: A => `{*}` IterableOnce[B]): `{it, f}` Iterator[B]
    def ++[B >: A](xs: `{*}` IterableOnce[B]): `{it, xs}` Iterator[B]
\end{lstlisting}
\noindent
Here, methods {\tt apply}, {\tt foldLeft}, {\tt foldRight}, {\tt foreach} again have
the same signatures as in the current Scala standard library. But the remaining
four operations need additional capture annotations. Method {\tt drop} on iterators
returns the given iterator {\tt it} after skipping {\tt n} elements. Consequently,
its result has {\tt \{it\}} as capture set. Methods {\tt map} and {\tt flatMap} lazily map
elements of the current iterator as the result is traversed. Consequently they
retain both {\tt it} and {\tt f} in their result capture set. Method {\tt ++} concatenates
two iterators and therefore retains both of them in its result capture set.

The examples attest to the practicality of capture checking. Method signatures are generally
concise. Higher-order methods over strict collections by and large keep the same types
as before. Capture annotations are only needed for capabilities that are retained in closures
and are executed on demand later. Where such annotations are needed they match the
developer's intuitive understanding of reference patterns and signal
information that looks relevant in this context.

\section{Related Work}\label{sec:related}

The results presented in this paper did not emerge in a vacuum and many of the underlying ideas
appeared individually elsewhere in similar or different form.
We follow the structure of the informal presentation in Section \ref{sec:informal}
and organize the discussion of related work according to the key ideas behind \CC{}.

\paragraph{\bf Effects as Capabilities}

Establishing effect safety by moderating access to effects via term-level capabilities is
not a new idea \cite{marino09generic}.
It has been proposed as a strategy to retrofit existing languages with means to
reason about effect safety \cite{choudhury2020recovering, liu2016, osvald2016gentrification}.
Recently, it also has been applied as the core principle behind a new programming language
featuring effect handlers \cite{brachthaeuser2020effects}.
Similar to the above prior work, we propose to use term-level capabilities to restrict
access to effect operations and other scoped resources with a limited lifetime.
Representing effects as capabilities results in a good economy of concepts:
existing language features, like term-level binders can be reused, programmers are not
confronted with a completely new concept of effects or regions.

\paragraph{\bf Making Capture Explicit}
Having a term-level representation of scoped capabilities introduces the
challenge to restrict use of such capabilities to the scope in which they are still live.
To address this issue, effect systems have been introduced \cite{zhang2019abstraction,biernacki2020binders, brachthaeuser2020effekt}
but those can result in overly verbose and difficult to understand types \cite{brachthaeuser2020effects}.
A third approach, which we follow in this paper, is to make capture explicit in the type of
functions.

\citet{hannan1998escape}
proposes a type-based escape analysis with the goal to facilitate stack allocation.
The analysis tracks variable reference using a type-and-effect system
and annotates every function type with the set of free variables it captures. The authors leave
the treatment of effect polymorphism to future work.
In a similar spirit, \citet{scherer2013tracking} present Open Closure Types to facilitate reasoning about
data flow properties such as non-interference. They present an extension of the simply typed lambda calculus
that enhances function types $[ \Gamma_0 ](\tau) \rightarrow \tau$ with the lexical environment $\Gamma_0$ that was originally
used to type the closure.

\citet{brachthaeuser2022effects} show System C, which mediates between first- and second-class values with boxes.
In their system, scoped capabilities are second-class values.
Normally, second-class values cannot be returned from any scope, but in System C they can be boxed and returned from \emph{some} scopes.
The type of a boxed second-class value tracks which scoped capabilities it has captured and accordingly, from which scopes it cannot be returned.
System C tracks second-class values with a coeffect-like environment and uses an effect-like discipline for tracking captured capabilities, which can in specific cases be more precise than $\cv$.
In comparison, \CC does not depend on a notion of second-class values and deeply integrates capture sets with subtyping.

Recently, \citet{bao2021reachability} have proposed to qualify types with
{\it reachability sets}. Their {\it reachability types} allow reasoning about
non-interference, scoping and uniqueness by tracking for each reference what
other references it may alias or (indirectly) point to.
Their system formalizes subtyping but not universal polymorphism.
However, it relates reachability sets along a different dimension
than \CC. Whereas in \CC a subtyping relationship is established between a capability $c$ and
the capabilities in the type of $c$, reachability types assume a subtyping relationship between
a variable $x$ and the variable owning the scope where $x$ is defined.
%
Reachability types track \comment{pedantic} detailed points-to and aliasing
information in a setting with mutable variables, while \CC is a more
foundational calculus for tracking references and capabilities that can be
used as a guide for an implementation in a complete programming language.
It would be interesting to explore how reachability types can be expressed in \CC.

\paragraph{\bf Capture Polymorphism}
Combining effect tracking with higher-order functions immediately gives rise to effect polymorphism,
which has been a long-studied problem.

Similar to the usual (parametric) type polymorphism, the seminal work by \citet{lucassenPolymorphicEffectSystems1988}
on type and effect systems featured (parametric) \emph{effect polymorphism} by adding language constructs for explicit
region abstraction and application. Similarly, work on region based memory management \cite{tofte1997region}
supports \emph{region polymorphism} by explicit region abstraction and application.
Recently, languages with support for algebraic effects and handlers, such as Koka \cite{leijen2017type}
and Frank \cite{lindleyBeBe2017}, feature explicit, parametric effect polymorphism.

It has been observed multiple times, for instance by \citet{osvald2016gentrification} and \citet{brachthaeuser2020effects},
that parametric effect polymorphism can become verbose and results in complicated types
and confusing error messages.
Languages sometimes attempt to \emph{hide} the complexity -- they ``simplify the types more and leave
out `obvious' polymorphism'' \cite{leijen2017type}. However, this solution is not satisfying since the
full types resurface in error messages.
In contrast, we support polymorphism by reusing existing term-level binders and support
simplifying types by means of subtyping and subcapturing.


%

The problem of how to prevent capabilities from escaping in closures is also addressed by
\emph{second-class values} that can only be passed as arguments but not be returned in results or stored
in mutable fields. \citet{siek2012effects} enforce second-class function arguments
using a classical polymorphic effect discipline whereas \citet{osvald2016gentrification} and \citet{brachthaeuser2020effects}
present a specialized type discipline for this task. Second-class values cannot be returned or closed-over by first-class functions.
On the other hand, second-class functions can
freely close over capabilities, since they are second-class themselves. This gives rise to a
convenient and light-weight form of \emph{contextual} effect polymorphism \cite{brachthaeuser2020effects}.
While this approach allows for effect polymorphism with a simple type system, it is also restrictive
because it also forbids local returns and retentions of capabilities; a problem solved by adding
boxing and unboxing \cite{brachthaeuser2022effects}.

\paragraph{\bf Foundations of Boxing}

%

Contextual modal type theory (CMTT) \cite{nanevski2008contextual} builds on
intuitionistic modal logic.
In intuitionistic modal logic, the graded propositional constructor $[ \Psi ]\;A$
(pronounced \emph{box}) witnesses that $A$ can be proven only using true propositions in $\Psi$.
Judgements in CMTT have two contexts: $\Gamma$, roughly corresponding to \CC
bindings with $\UC$ as their capture set, and a modal context $\Delta$
roughly corresponding to bindings with concrete capture sets.
Bindings in the modal context are necessarily boxed and annotated with a modality $ x :: A[\Psi] \in \Delta$.
Just like our definition of captured variables in \CC{}, the definition of free variables by \citet{nanevski2008contextual}
assigns the empty set to a boxed term (that is, $fv(\mathsf{box}(\Psi{}. M)) = \{\}$).
Similar to our unboxing construct, using a variable bound in the modal context requires that
the current context satisfies the modality $\Psi$, mediated by a substitution $\sigma$.
Different to CMTT, \CC{} does not introduce a separate modal context. It also does not annotate modalities on binders, instead these are kept in the types.
Also different to CMTT, in \CC{} unboxing is annotated with a capture set and not a substitution.

Comonadic type systems were introduced to support reasoning about \emph{purity}
in existing, impure languages \cite{choudhury2020recovering}.
Very similar to the box modality of CMTT,
a type constructor `Safe' witnesses the fact that its values
are constructed without using any impure capabilities.
The type system presented by \citet{choudhury2020recovering} only supports a binary
distinction between \emph{pure} values and \emph{impure} values, however, the authors
comment that it might be possible to generalize their system to graded modalities.

In the present paper, we use boxing as a practical tool, necessary to obtain concise types when combining
capture tracking with parametric type polymorphism.

\paragraph{\bf Coeffect Systems}

{\it Coeffect systems} also attach additional information to bindings in the
environment, leading to a typing judgment of the form $\Gamma @ \;
\mathcal{C} \ts e \typ \tau$. Such systems can be seen as similar in spirit to
\CC, where additional information is available about each variable in the
environment through the capture set of its type.
\citet{petricek2014coeffects} show a general coeffect framework that can
be instantiated to track various concepts such as bounded reuse of variables,
implicit parameters and data access.
This framework is based on simply typed lambda calculus and
its function types are always coeffect-monomorphic.
In contrast, \CC is based on \sysfsub (thus supporting type polymorphism and
subtyping) and supports capture-polymorphic functions.

\paragraph{\bf Object Capabilities}

The (object-)capability model of programming \cite{crary1999typed,boyland2001capabilities,miller2006robust},
controls security critical operations by requiring access to a capability.
Such a capability can be seen as the constructive proof that the holder
is entitled to perform the critical operation.
Reasoning about which operations a module can perform is reduced
to reasoning about which references to capabilities a module holds.

The Wyvern language \cite{melicher2017capability} implements this model
by distinguishing between stateful \emph{resource modules} and \emph{pure modules}.
Access to resource modules is restricted and only possible through capabilities.
Determining the authority granted by a module amounts to manually inspecting its type signature and all of the type signatures of its
transitive imports. To support this analysis, \citet{melicher2020controlling} extends the language with a
fine-grained effect system that tracks access of capabilities in the type of methods.

In \CC, one can reason about authority and capabilities simply by inspecting
capture sets of types.
If we model modules via function abstraction, the function's capture set directly reflects its authority.
Importantly, \CC does not include an effect system and thus tracks mention
rather than use.

\section{Conclusion}\label{sec:conc}

We introduced a new type system \CC to track captured references of values. Tracked references
are restricted to capabilities, where capabilities are references bootstrapped
from other capabilities, starting with the universal capability. Implementing this simple principle
then naturally suggests a chain of design decisions:
\begin{enumerate}
  \item Because capabilities are variables, every function must have its type annotated with its free capability variables.
  \item To manage the scoping of those free variables, function types must be dependently-typed.
  \item To prevent non-variable terms from occurring in types, the programming language is formulated in monadic normal form.
  \item Because of type dependency, the let-bindings of MNF have to satisfy the avoidance property,
        to prevent out-of-scope variables from occurring in types.
  \item To make avoidance possible, the language needs a rich notion of subtyping on the capture sets.
  \item Because the capture sets represent object capabilities, the subcapture relation cannot just be the subset relation on sets of variables -- it also has to take into account the types of the variables,
        since the variables may be bound to values which themselves capture capabilities.
  \item To keep the size of the capture sets from ballooning out of control, the paper introduces a box connective with box and unbox rules to control when free variables are counted as visible.
\end{enumerate}
\noindent
We showed that the resulting system can be used as the basis for lightweight polymorphic effect checking,
without the need for effect quantifiers. We also identified three key principles that
keep notational overhead for capture tracking low:
\begin{itemize}
\item[--] Variables are tracked only if their types have non-empty capture sets.
In practice the majority of variables are untracked and thus do not need to be mentioned at all.
\item[--] Subcapturing, subtyping and subsumption mean that more detailed capture sets
   can be subsumed by coarser ones.
\item[--] Boxed types stop propagation of capture information in enclosing types which avoids
  repetition in capture annotations to a large degree.
\end{itemize}
Our experience so far indicates that the presented calculus is simple and expressive enough to be used as a basis for more
advanced effect and resource checking systems and their practical implementations.


\begin{acks}
  Acknowledgments. We thank Jonathan Aldrich, Vikraman Choudhury and Neal Krishnawami for their input
  in discussions about this research. We thank the anonymous reviewers of previous versions of this paper for their
  comments and suggestions. In particular, we paraphrased with permission one reviewer's summary
  of our design decisions in the conclusion.
  This research was partially funded by the Natural Sciences and Engineering Research Council of Canada.
\end{acks}

\bibliography{bibliography}

\appendix


\input{proof.tex}

\end{document}

%% file: cc-macros.tex
\newcommand{\note}[2][red]{}
\newcommand{\todo}[1]{\note{TODO: #1}}
\renewcommand{\comment}[1]{}

\newcommand{\sysfsub}{$\textsf{System F}_{<:}~$}
\newcommand{\fsub}{$\textsf{F}_{<:}~$}
\newcommand\CC{$\textsf{CC}_{<:\Box}~$}

\newcounter{RuleRef}
\newcommand{\ruledef}[1]{%
  \refstepcounter{RuleRef}\textsc{#1}\label{rule:#1}}

\newcommand{\ruleref}[1]{\textsc{(\hyperref[rule:#1]{#1})}}
\newcommand{\rulerefN}[2]{\textsc{(\hyperref[rule:#1]{#2})}}

\newcommand{\judgement}[2]{{\bf\textsf{#1}} \hfill #2}

\newcommand\kw[1]{\operatorname{\textbf{\textsf{#1}}}}

\definecolor{light-gray}{gray}{0.92}

\newcommand{\greyed}[1]{\mbox{\colorbox{light-gray}{$#1$}}}

\newcommand{\rcap}{\mbox{$*$}}

\newcommand\cref[1]{\id{\textbf{ref}}_{#1}}

\renewcommand\Top{\top}

\newcommand\Unit{\id{Unit}}

\newcommand\tlam[3]{\Lambda\left[#1 <: #2\right]\!.\; #3}

\newcommand\LAM[1]{\lambda\left(#1\right)\;}
\newcommand\TLAM[1]{\Lambda\left[#1\right]\;}
\newcommand\ARR[1]{\forall(#1)}
\newcommand\TARR[1]{\forall[#1]}

\newcommand\Boxed[1]{\Box~#1}
\newcommand\Unboxed[2]{#1 \multimapinv #2}
\newcommand\Let[3]{\kw{let} #1 = #2 \kw{in} #3}

\newcommand\Ttlam[3]{\forall\left[#1 <: #2\right] \rightarrow #3}
\newcommand\dom[1]{\operatorname{dom} #1}

\newcommand\cset[1]{\{#1\}}
\newcommand\UC{\cset{\rcap}}

\newcommand\CAPT[2]{\cset{#1} \; #2}
\newcommand\CS[2]{#1 \; #2}
\newcommand\Capt[2]{#1 \; #2}

\newcommand{\OR}{\mathop{\ \ \ |\ \ \ }}
\newcommand{\ts}{\,\vdash\,}
\newcommand\sub{<:}

\newcommand\typ{:}

\newcommand\fv{\operatorname{fv}}
\newcommand\cv{\operatorname{cv}}

\newcommand\id[1]{\operatorname{\mathsf{#1}}}

\newcommand{\IF}{\kw{if}}

\newcommand\wf{\;\mbox{$\kw{wf}$}}

\newcommand{\reduces}{\longrightarrow}

\newcommand\storeof[1]{\sigma[\ #1\ ]}
\newcommand\evalof[1]{e[\ #1\ ]}

\newcommand{\gap}{\quad\quad}

\newcommand{\bbox}[1]{\mbox{\textbf{\textsf{#1}}}}

%% file: syntax.tex
\begin{figure*}[h]
  \vspace{0.3em}
  \begin{center}\rule{0.9\textwidth}{0.4pt}\end{center}
  \vspace{0.3em}

  $\begin{array}[t]{llll@{\hspace{8mm}}l}
    \bbox{Variable} & \multicolumn{2}{l}{x, y, z} \\
    \bbox{Type Variable} & \multicolumn{2}{l}{X, Y, Z} \\[1em]
    \bbox{Value} & v, w & ::= & \lambda(x : T)t
        \OR \lambda[X \sub S]t
        \OR \mbox{\colorbox{light-gray}{$\Boxed{x}$}} \\

    \bbox{Answer} & a & ::= & v \OR x \\

    \bbox{Term} &s, t& ::= & a
          \OR  x\,y
          \OR  x\,[S]
          \OR  \Let x s t
          \OR  \greyed{\Unboxed{C}{x}}\\

    \bbox{Shape Type} & S & ::= & X
          \OR  \Top
          \OR  \forall(x : U) T
          \OR  \forall[X \sub S] T
          \OR  \greyed{\Boxed{T}} \\

    \bbox{Type} & T, U & ::= & S \OR \greyed{\Capt{C}{S}} \\

    \bbox{Capture Set} & \greyed{C} & ::= & \greyed{\{x_1, \ldots, x_n\}}
    \end{array}$

\caption{\label{fig:syntax} Syntax of System \CC}
\vspace{0.3em}
\begin{center}\rule{0.9\textwidth}{0.4pt}\end{center}
\vspace{0.3em}
\end{figure*}

%% file: cc-rules.tex
\begin{figure*}
\small

  \judgement{Subcapturing}{\fbox{${\Gamma \ts C \sub C}$}}

  \vspace{0.4em}

  \begin{minipage}{0.3\textwidth}
    \infrule[\ruledef{sc-elem}]{\greyed{x \in C}}{\greyed{\Gamma \ts \cset{x} \sub C}}
  \end{minipage}
  \begin{minipage}{0.5\textwidth}
  \infrule[\ruledef{sc-set}]{%
    \greyed{\Gamma \ts \cset{x_1} \sub C \ \ldots\ \Gamma \ts \cset{x_n} \sub C}
  }{%
    \greyed{\Gamma \ts \cset{x_1, \ldots, x_n} \sub C}}
  \end{minipage}

  \vspace{0.4em}

  \begin{minipage}{0.35\textwidth}
  \infrule[\ruledef{sc-var}]{%
    \greyed{x : C'\;S \in \Gamma \quad
    \Gamma \ts C' \sub C}
   }{%
     \greyed{\Gamma \ts \cset{x} \sub C}}
  \end{minipage}


\vspace{0.3em}
\begin{center}\rule{0.9\textwidth}{0.4pt}\end{center}
\vspace{0.3em}

\judgement{Subtyping}{\fbox{$\Gamma \ts T \sub T$}}
\vspace{0.8em}

  \begin{minipage}{0.5\textwidth}
  \infax[\ruledef{refl}]{%
    \Gamma \ts T \sub T
  }

  \vspace{0.8em}
  \infrule[\ruledef{tvar}]{%
    X \sub S \in \Gamma
  }{%
    \Gamma \ts X \sub S}

  \vspace{0.8em}
  \infrule[\ruledef{fun}]{%
    \Gamma \ts U_2 \sub U_1 \gap
    \Gamma, x: U_2 \ts T_1 \sub T_2
  }{%
    \Gamma \ts \forall(x: U_1)T_1 \sub \forall(x: U_2)T_2}

  \vspace{0.8em}
  \infrule[\ruledef{capt}]{%
    \greyed{\Gamma \ts C_1 \sub C_2 \gap
    \Gamma \ts S_1 \sub S_2}
  }{%
    \greyed{\Gamma \ts C_1\;S_1 \sub C_2\;S_2}}
  \end{minipage}\begin{minipage}{0.5\textwidth}

  \infrule[\ruledef{trans}]{%
    \Gamma \ts T_1 \sub T_2 \gap \Gamma \ts T_2 \sub T_3
  }{%
    \Gamma \ts T_1 \sub T_3
  }
  \vspace{0.8em}

  \infax[\ruledef{top}]{%
    \Gamma \ts S \sub \top}

  \vspace{0.8em}
  \infrule[\ruledef{tfun}]{%
    \Gamma \ts S_2 \sub S_1 \gap
    \Gamma, x: S_2 \ts T_1 \sub T_2
  }{%
    \Gamma \ts \forall[X \sub S_1]T_1 \sub \forall[X \sub S_2]T_2}

  \vspace{0.8em}
  \infrule[\ruledef{boxed}]{%
    \greyed{\Gamma \ts T_1 \sub T_2}
  }{%
    \greyed{\Gamma \ts \Boxed{T_1} \sub \Boxed{T_2}}}

  \end{minipage}


\vspace{0.3em}
\begin{center}\rule{0.9\textwidth}{0.4pt}\end{center}
\vspace{0.3em}

  \judgement{\textsf{Typing}}{\fbox{$\Gamma \ts t \typ T$}}

  \vspace{0.8em}
  \begin{minipage}{0.5\textwidth}
  \infrule[\ruledef{var}]{%
    x: {\greyed C}{S} \in \Gamma
  }{%
    {\Gamma \ts x \typ \greyed{\{x\}} S}}

  \vspace{0.8em}
  \infrule[\ruledef{abs}]{%
    \Gamma, x: U \ts t: T
    \gap
    \Gamma \ts U \wf
  }{%
    \Gamma \ts \lambda(x: U)t \typ \greyed{\cv(t)\backslash x}{\forall(x: U)T}}

  \vspace{0.8em}
  \infrule[\ruledef{app}]{%
    \Gamma \ts x \typ \greyed C \forall(z: U)T \gap
    \Gamma \ts y \typ U
  }{%
    \Gamma \ts x\,y \typ [z := y]T}

  \vspace{0.8em}
  \infrule[\ruledef{box}]{%
    \greyed{\Gamma \ts x \typ C\;S \gap
    C \subseteq \dom(\Gamma)}
  }{%
    \greyed{\Gamma \ts \Box~x \typ \Box~C\;S}}

  \end{minipage}\begin{minipage}{0.5\textwidth}
  \infrule[\ruledef{sub}]{%
    \Gamma \ts t \typ T
    \gap
    \Gamma \ts T \sub U
    \gap
    \Gamma \ts U \wf
  }{%
    \Gamma \ts t \typ U }

  \vspace{0.8em}
  \infrule[\ruledef{tabs}]{%
    \Gamma, X \sub S \ts t: T
    \gap
    \Gamma \ts S \wf
  }{%
    \Gamma \ts \lambda[X \sub S]t \typ \greyed{\cv(t)} \forall[X \sub S]T}

  \vspace{0.8em}
  \infrule[\ruledef{tapp}]{%
    \Gamma \ts x \typ \greyed C \forall[X \sub S]T
  }{%
    \Gamma \ts x\,[S] \typ [X := S]T}

  \vspace{0.8em}
  \infrule[\ruledef{unbox}]{%
    \greyed{\Gamma \ts x \typ \Box~C\;S \gap
    C \subseteq \dom(\Gamma)}
  }{%
    \greyed{\Gamma \ts \Unboxed C x \typ C\;S}}
  \end{minipage}

  \vspace{0.8em}
  \begin{minipage}{0.5\textwidth}
    \infrule[\ruledef{let}]{%
    \Gamma \ts s: T \gap \Gamma, x: T \ts t: U \gap x \notin \fv(U)
  }{%
    \Gamma \ts \Let x s t : U}
  \end{minipage}

  \vspace{0.3em}
  \begin{center}\rule{0.9\textwidth}{0.4pt}\end{center}
  \vspace{0.3em}

  \judgement{Evaluation}{\fbox{$\Gamma \ts t \reduces t'$}}
  \vspace{0.8em}
  \[
  \begin{array}{lclll}
    \storeof{\evalof{x\,y}}  &\reduces& \storeof{\evalof{[z := y]t}} & \gap \IF\ \sigma(x) = \lambda(z: T)t & \mbox{\sc (apply)} \\
    \storeof{\evalof{x\,[S]}}   &\reduces& \storeof{\evalof{[X := S]t}} & \gap \IF\ \sigma(x) = \lambda[X<:S']t & \mbox{\sc (tapply)} \\
    \storeof{\evalof{\greyed{\Unboxed{C}{x}}}} &\reduces& \storeof{\evalof{\greyed y}} & \gap \IF\ \sigma(x) = \greyed{\Boxed{y}} & \mbox{\sc (open)} \\
    \storeof{\evalof{\Let x y t}} &\reduces& \storeof{\evalof{[x := y]t}} && \mbox{\sc (rename)} \\
    \storeof{\evalof{\Let x v t}} &\reduces& \storeof{\Let x v {\evalof{t}}} & \gap \IF\ e \neq [\,] & \mbox{\sc (lift)} \\
  \end{array}
  \]

  \begin{minipage}{\textwidth}
  {\bf \ \ \ where}
  $\begin{array}[t]{llll@{\hspace{8mm}}l}
  \bbox{Store context} & \sigma& ::= & \multicolumn{2}{l}{[\,] \OR \Let x v \sigma} \\
  \bbox{Eval context}  & e &::=& \multicolumn{2}{l}{[\,] \OR \Let x e t}
  \end{array}$
  \end{minipage}

\caption{\label{fig:cc-rules} Typing and Evaluation Rules of System \CC}
\end{figure*}

%% file: proof.tex







\newenvironment{proofcases}{\begin{adjustwidth}{1cm}{}}{\end{adjustwidth}}
\newcommand{\pcase}[1]{\noindent{\it Case #1.}}

\newcommand{\jona}[1]{\note[gray]{Jonathan: #1} }
\newcommand{\abg}[1]{\note[gray]{ABG: #1} }

\newtheorem{fact}{Fact}[section]

\newcommand{\eruleref}[1]{{\sc(#1)}}
\newcommand{\lref}[1]{Lemma \ref{#1}}





\section{Proofs}

\subsection{Proof devices}

\note{The paper should already define well-formedness of types.}

We extend type well-formedness to environments as $\ts \Gamma \wf$ in the
typical way -- empty environment is well-formed, and $\Gamma,x:T$ or
$\Gamma,X<:T$ is well-formed if both $\ts \Gamma \wf$ and $\Gamma \ts T \wf$.

We implicitly assume all variables in an environment are unique, for instance
whenever we write $\Gamma,x:T$ we assume that $x \not\in \dom(\Gamma)$.

In order to make proving Preservation (\autoref{th:preservation}) possible,
we need to relate the derivation of a type for a term of the form $\storeof{t}$
to the derivation of the type of the {\it plug} term $t$ inside the store $\sigma$.
To do so, we need to define what environment {\it matches} the store.

\begin{figure}[H]
  \judgement{\textsf{Matching environment}}{\fbox{$\Gamma \ts \sigma \sim \Delta$}}\\
  [1em]
  \begin{minipage}{0.5\textwidth}
    \infrule{%
      \Gamma,x:T \ts \sigma \sim \Delta \quad
      \Gamma \ts v \typ T \quad
      x \not\in \fv(T)
    }{%
      \Gamma \ts \Let{x}{v}{\sigma} \sim x:T,\Delta
    }
  \end{minipage}%
  \begin{minipage}{0.5\textwidth}
    \infax{%
      \Gamma \ts [\,] \sim \cdot
    }
  \end{minipage}
\end{figure}

\begin{definition}[Evaluation context typing ($\Gamma |- e : U => T$)]
  We say that $\Gamma$ types an evaluation context $e$ as $U => T$ if $\Gamma,x:U |- \evalof{x} : T$.
\end{definition}

\FloatBarrier

\begin{fact}
  If $\storeof{t}$ is a well-typed term in $\Gamma$,
  then there exists a $\Delta$ matching $\sigma$
  (i.e. such that $\Gamma \ts \sigma \sim \Delta$),
  finding it is decidable,
  and $\Gamma,\Delta$ is well-formed.
\end{fact}

\begin{fact}
  The analogous holds for $\evalof{t}$.
\end{fact}

\subsection{Properties of Evaluation Contexts and Stores}

In the proof, we will be using $D$ as another metavariable for capture sets.

We will also be writing $\cv(T)$ to extract the capture set of a type (which is
defined simply as $\cv(C\;R) = C$).

We assume that all environments in the premises of our lemmas and theorems are
well-formed.

\begin{lemma}[Inversion of typing under an evaluation context]
  \label{lm:ectx-plug-typ}
  If $\Gamma \ts \evalof{u} : T$, then $u$ is well-typed in $\Gamma$.
\end{lemma}

\begin{proof}
  Proceed by induction on $e$.  If $e$ is empty then the result directly holds,
  so consider when $s = \Let{x}{e'[u]}{s}$.  By inverting the typing derivation
  we know that $\Gamma \ts e'[u] : T'$, and hence by induction we conclude that
  $\Gamma \ts s : U$ for some type $U$.
\end{proof}

\begin{lemma}[Evaluation context typing inversion]\forcenl
  \indent
  $\Gamma \ts \evalof{s} : T$ implies that for some $U$ we have $\Gamma |- e : U => T$ and $\Gamma |- s : U$.
\end{lemma}
\begin{proof}
  First, we show that for some $U$, we have $\Gamma,x:U |- \storeof{e} : T$.
  This can be done by a straightforward induction on the derivation of typing for $\evalof{s}$.

  By definition, this then gives us $\Gamma |- e : U => T$.
\end{proof}

\begin{lemma}[Evaluation context reification]\forcenl
  \label{lm:tm-rectx-eval}
  \indent
  If both $\Gamma \ts e : U => T$ and $\Gamma \ts s : U$, then $\Gamma \ts \evalof{s} : T$.
\end{lemma}
\begin{proof}
  A corollary of Lemma~\label{lm:tm-subst-typ}.
\end{proof}

\begin{lemma}[Replacement of term under an evaluation context]
  \label{lm:ectx-replug}
  If $\Gamma \ts \evalof{s} : U => T$
  and $\Gamma \ts s' : U$,
  then $\Gamma \ts \evalof{s'} : T$.
\end{lemma}
\begin{proof}
  A corollary of \label{lm:tm-rectx-eval}.
\end{proof}

\begin{lemma}[Store context reification]
  \label{lm:tm-rectx-store}
  If $\Gamma,\Delta \ts t \typ T$
  and $\Gamma \ts \sigma \sim \Delta$
  then $\Gamma \ts \storeof{t} \typ T$.
\end{lemma}

\begin{proof}
  By induction on $\sigma$.

  \begin{proofcases}
    \pcase{$\sigma = [\,]$} Immediate.

    \pcase{$\sigma = \sigma'[\ \Let{x}{v}{[\,]}\ ]$}
    Then $\Delta = \Delta', x:U$ for some $U$.
    Since $x \not\in \fv(T)$ as $\Gamma \ts T \wf$,
    by \ruleref{let}, we have that $\Gamma, \Delta' \ts \Let{x}{v}{t}$
    and hence by the induction hypothesis we have, for some $U$, that
    $\Gamma, x \typ U \ts \sigma'[t] \typ T$.  The result follows directly.
  \end{proofcases}
\end{proof}

The above lemma immediately gives us:

\begin{corollary}[Replacement of term under a store context]
  \label{lm:sctx-replug}
  If $\Gamma \ts \storeof{t} : T$
  and $\Gamma \ts \sigma \sim \Delta$
  and $\Gamma,\Delta \ts t : T$,
  then for all $t'$
  such that $\Gamma,\Delta \ts t' : T$
  we have $\Gamma \ts \storeof{t'} : T$.
\end{corollary}

\subsection{Properties of Subcapturing}

\begin{lemma}[Top capture set]\label{lm:top-cs}
  Let $\Gamma |- C \wf$. Then $\Gamma |- C <: \UC$.
\end{lemma}
\begin{proof}
  By induction on $\Gamma$.
  If $\Gamma$ is empty, then $C$ is either empty or $\rcap \in C$, so we can conclude by \ruleref{sc-set} or \ruleref{sc-elem} correspondingly.
  Otherwise, $\Gamma = \Gamma',x:D\,S$ and since $\Gamma$ is well-formed, $\Gamma' |- D \wf$.
  By \ruleref{sc-set}, we can conclude if for all $y \in C$ we have $\Gamma |- \{y\} <: \UC$.
  If $y = x$, by IH we derive $\Gamma' |- D <: \UC$, which we then weaken to $\Gamma$ and conclude by \ruleref{sc-var}.
  If $y \neq x$, then $\Gamma' |- \{y\} \wf$, so by IH we derive $\Gamma' |- \{y\} <: \UC$ and conclude by weakening.
\end{proof}

\begin{corollary}[Effectively top capture set]\label{lm:effectitop-cs}
  Let $\Gamma |- C,D \wf$ such that $\rcap \in D$. Then we can derive $\Gamma |- C <: D$.
\end{corollary}
\begin{proof}
  We can derive $\Gamma |- C <: \UC$ by Lemma~\ref{lm:top-cs} and then we can conclude by Lemma~\ref{lm:sc-trans} and \ruleref{sc-elem}.
\end{proof}

\begin{lemma}[Star subcapture inversion]\label{lm:star-sc-inversion}
  Let $\Gamma |- C <: D$. If $\rcap \in C$, then $\rcap \in D$.
\end{lemma}
\begin{proof}
  By induction on subcapturing.
  Case \ruleref{sc-elem} immediate, case \ruleref{sc-set} by repeated IH, case \ruleref{sc-var} contradictory.
\end{proof}

\begin{lemma}[Subcapturing distributivity]
  \label{lm:sc-distr}
  If $\Gamma \ts C \sub D$ and $x \in C$,
  then $\Gamma \ts \{x\} \sub D$.
\end{lemma}

\begin{proof}
  By inspection of the last subcapturing rule used to derive $C \sub D$.
  All cases are immediate. If the last rule was \ruleref{sc-set}, we have our goal
  as premise. Otherwise, we have that $C = \{x\}$, and the result
  follows directly.
\end{proof}

\begin{lemma}[Subcapturing transitivity]
  \label{lm:sc-trans}
  If $\Gamma \ts C_1 \sub C_2$
  and $\Gamma \ts C_2 \sub C_3$
  then $\Gamma \ts C_1 \sub C_3$.
\end{lemma}

\begin{proof}
  By induction on the first derivation.

  \begin{proofcases}
    \pcase{\ruleref{sc-elem}}
    $C_1 = \{x\}$ and $x \in C_2$,
    so by \lref{lm:sc-distr} $\Gamma \ts \{x\} \sub C_3$.

    \pcase{\ruleref{sc-var}}
    Then $C_1 = \{x\}$ and $x : C_4\;R \in \Gamma$ and $\Gamma \ts C_4 \sub C_2$.
    By IH $\Gamma \ts C_4 \sub C_3$
    and we can conclude by \ruleref{sc-var}.

    \pcase{\ruleref{sc-set}}
    By repeated IH and \ruleref{sc-set}.
  \end{proofcases}
\end{proof}

\begin{lemma}[Subcapturing reflexivity]
  \label{lm:sc-refl}
  For all $C$, if $\Gamma \ts C \wf$, then $\Gamma \ts C \sub C$.
\end{lemma}

\begin{proof}
  By \ruleref{sc-set} and \ruleref{sc-elem}.
\end{proof}

\begin{lemma}[Subtyping implies subcapturing]
  \label{lm:sub-is-sc}
  If $\Gamma \ts C_1\;R_1 \sub C_2\;R_2$,
  then $\Gamma \ts C_1 \sub C_2$.
\end{lemma}

\begin{proof}
  By induction on the subtyping derivation.
  If \ruleref{capt}, immediate.
  If \ruleref{trans}, by IH and subcapturing transitivity \lref{lm:sc-trans}.
  Otherwise, $C_1 = C_2 = \{\}$ and we can conclude by \ruleref{sc-set}.
\end{proof}

\subsubsection{Subtype inversion}

While it may be unusual to show these lemmas so early, we in fact need one of
them to show narrowing.

\begin{fact}
  \label{fact:trans}
  Both subtyping and subcapturing are transitive.
\end{fact}

\begin{proof}
  Subtyping is intrisically transitive through \ruleref{trans},
  while subcapturing admits transitivity as per \lref{lm:sc-trans}.
\end{proof}

\begin{fact}
  \label{fact:refl}
  Both subtyping and subcapturing are reflexive.
\end{fact}

\begin{proof}
  Again, this is an intrinsic property of subtyping (by \ruleref{refl})
  and an admissible property of subcapturing per \lref{lm:sc-refl}.
\end{proof}

\begin{lemma}[Subtyping inversion: type variable]
  \label{lm:sub-inv-tvar}
  If $\Gamma \ts U \sub \CS{C}X$,
  then $U$ is of the form $\CS{C'}X'$
  and we have $\Gamma \ts C' \sub C$
  and $\Gamma \ts X' \sub X$.
\end{lemma}

\begin{proof}
  By induction on the subtyping derivation.
  \begin{proofcases}
    \pcase{\ruleref{tvar}, \ruleref{refl}} Follow from reflexivity (\ref{fact:refl}).

    \pcase{\ruleref{capt}}
    Then we have $U = \CS{C'}R$ and $\Gamma \ts C' \sub C$
    and $\Gamma \ts R \sub X$.\\
    \indent This relationship is equivalent to
    $\Gamma \ts \CAPT{}R \sub \CAPT{}X$,
    on which we invoke the IH.\\
    \indent By IH we have $\CAPT{}R = \CAPT{}Y$
    and we can conclude with $U = \CS{C'}Y$.

    \pcase{\ruleref{trans}}
    Then we have $\Gamma \ts U \sub U$
    and $\Gamma \ts U \sub \CS{C}X$.
    We proceed by using the IH twice
    and conclude by transitivity (\ref{fact:trans}).

    \noindent Other rules are impossible.
  \end{proofcases}
\end{proof}

\begin{lemma}[Subtyping inversion: capturing type]
  \label{lm:sub-inv-capt}
  If $\Gamma \ts U \sub \CS{C}R$,
  then $U$ is of the form $\CS{C'}R'$
  such that $\Gamma \ts C' \sub C$
  and $\Gamma \ts R' \sub R$.
\end{lemma}

\begin{proof}
  We take note of the fact that subtyping and subcapturing
  are both transitive (\ref{fact:trans})
  and reflexive (\ref{fact:refl}).
  The lemma then follows from a straightforward induction on the subtyping
  judgment.
\end{proof}

\begin{lemma}[Subtyping inversion: function type]
  \label{lm:sub-inv-fun}
  If $\Gamma \ts U \sub \CS{C}\ARR{x:T_1}T_2$,
  then $U$ either is of the form $\CS{C'}X$
  and we have $\Gamma \ts C' \sub C$
  and $\Gamma \ts X \sub \ARR{x:T_1}T_2$,
  or $U$ is of the form $\CS{C'}\ARR{x:U_1}U_2$
  and we have $\Gamma \ts C' \sub C$
  and $\Gamma \ts T_1 \sub U_1$
  and $\Gamma,x:T_1 \ts U_2 \sub T_2$.
\end{lemma}

\begin{proof}
  By induction on the subtyping derivation.
  \begin{proofcases}
    \pcase{\ruleref{tvar}} Immediate.

    \pcase{\ruleref{fun}, \ruleref{refl}} Follow from reflexivity (\ref{fact:refl}).

    \pcase{\ruleref{capt}}
    Then we have $\Gamma \ts C' \sub C$
    and $\Gamma \ts R \sub \ARR{x:T_1}T_2$.\\
    \indent This relationship is equivalent to
    $\Gamma \ts \CAPT{}R \sub \CAPT{}\ARR{x:T_1}T_2$,
    on which we invoke the IH.\\
    \indent By IH $\CAPT{}R$ might have two forms.
    If $\CAPT{}R = \CAPT{}X$, then we can conclude with $U = \CS{C'}X$.\\
    \indent Otherwise we have $\CAPT{}R = \CAPT{}\ARR{x:U_1}U_2$
    and $\Gamma \ts T_1 \sub U_1$
    and $\Gamma,x:T_1 \ts U_2 \sub T_2$.
    Then, $U = \CS{C'}\ARR{x:U_1}U_2$ lets us conclude.

    \pcase{\ruleref{trans}}
    Then we have $\Gamma \ts U \sub U'$
    and $\Gamma \ts U \sub \CS{C}\ARR{x:T_1}T_2$.
    By IH $U$ may have one of two forms.
    If $U = \CS{C'}X$,
    we proceed with \lref{lm:sub-inv-tvar}
    and conclude by transitivity (\ref{fact:trans}).\\
    \indent
    Otherwise $U = \CS{C'}\ARR{x:U_1}U_2$
    and we use the IH again on $\Gamma \ts U' \sub \CS{C'}\ARR{x:U_1}U_2$.
    If $U = \CS{C''}X$,
    we again can conclude by (\ref{fact:trans}).
    Otherwise if $U = \CS{C''}\ARR{x:U_1}U_2$,
    the IH only gives us $\Gamma,x:U_1 \ts U_2 \sub U_2$,
    which we need to narrow to $\Gamma,x:T_1$
    before we can similarly conclude by transitivity (\ref{fact:trans}).

    \noindent Other rules are not possible.
  \end{proofcases}
\end{proof}

\begin{lemma}[Subtyping inversion: type function type]
  \label{lm:sub-inv-tfun}
  If $\Gamma \ts U \sub \CS{C}\TARR{X:T_1}T_2$,
  then $U$ either is of the form $\CS{C'}X$
  and we have $\Gamma \ts C' \sub C$
  and $\Gamma \ts X \sub \TARR{X:T_1}T_2$,
  or $U$ is of the form $\CS{C'}\TARR{X:U_1}U_2$
  and we have $\Gamma \ts C' \sub C$
  and $\Gamma \ts T_1 \sub U_1$
  and $\Gamma,X<:T_1 \ts U_2 \sub T_2$.
\end{lemma}

\begin{proof}
  Analogous to the proof of \lref{lm:sub-inv-fun}.
\end{proof}

\begin{lemma}[Subtyping inversion: boxed type]
  If $\Gamma \ts U \sub \CS{C}\Boxed{T}$,
  then $U$ either is of the form $\CS{C'}X$
  and we have $\Gamma \ts C' \sub C$
  and $\Gamma \ts X \sub \Boxed{T}$,
  or $U$ is of the form $\CS{C'}\Boxed{U'}$
  and we have $\Gamma \ts C' \sub C$
  and $\Gamma \ts U \sub T$.
\end{lemma}

\begin{proof}
  Analogous to the proof of \lref{lm:sub-inv-fun}.
\end{proof}

\subsubsection{Permutation, weakening, narrowing}

\begin{lemma}[Permutation]
  Permutating the bindings in the environment up to preserving environment
  well-formedness also preserves type well-formedness, subcapturing, subtyping
  and typing.

  Let $\Gamma$ and $\Delta$ be the original and permutated context,
  respectively. Then:
  \begin{enumerate}
  \item If $\Gamma \ts T \wf$, then $\Delta \ts T \wf$.
  \item If $\Gamma \ts C_1 \sub C_2$, then $\Delta \ts C_1 \sub C_2$.
  \item If $\Gamma \ts U \sub T$, then $\Delta \ts U \sub T$.
  \item If $\Gamma \ts t \typ T$, then $\Delta \ts t \typ T$.
  \end{enumerate}
\end{lemma}

\begin{proof}
  As usual, order of the bindings in the environment is not used in any rule.
\end{proof}

\note{In fact, arbitrary permutation preserves all the above jdgmnts, but it
  might violate environment well-formedness, and we never want to do that.}

\begin{lemma}[Weakening]
  Adding a binding to the environment such that the resulting environment is
  well-formed preserves type well-formedness, subcapturing, subtyping and
  typing.

  Let $\Gamma$ and $\Delta$ be the original and extended context, respectively.
  Then:
  \begin{enumerate}
  \item If $\Gamma \ts T \wf$, then $\Delta \ts T \wf$.
  \item If $\Gamma \ts C_1 \sub C_2$, then $\Delta \ts C_1 \sub C_2$.
  \item If $\Gamma \ts U \sub T$, then $\Delta \ts U \sub T$.
  \item If $\Gamma \ts t \typ T$, then $\Delta \ts t \typ T$.
  \end{enumerate}
\end{lemma}

\begin{proof}
  As usual, the rules only check if a variable is bound in the environment
  and all versions of the lemma are provable by straightforward induction.
  For rules which extend the environment, such as \ruleref{abs}, we need
  permutation. All cases are analogous, so we will illustrate only one.

  \begin{proofcases}
    \pcase{\ruleref{abs}}
    WLOG we assume that $\Delta = \Gamma,x:T$.
    We know that $\Gamma \ts \LAM{y:U}t' \typ \ARR{y:U}U$.
    and from the premise of \ruleref{abs} we also know that
    $\Gamma,y:U \ts t' \typ U$.

    By IH, we have $\Gamma,y:U,x:T \ts t' \typ U$.
    $\Gamma,x:T,y:U$ is still a well-formed environment
    (as $T$ cannot mention $y$) and by permutation we have
    $\Gamma,x:T,y:U \ts t' \typ U$. Then by \ruleref{abs} we have
    $\Gamma,x:T \ts \LAM{y:U}t' \typ \ARR{y:U}U$,
    which concludes.
  \end{proofcases}
\end{proof}

\begin{lemma}[Type binding narrowing]\ \\
  \begin{enumerate}
  \item If $\Gamma \ts S' \sub S$
    and $\Gamma,X<:S,\Delta \ts T \wf$,
    then $\Gamma,X<:S',\Delta \ts T \wf$.

  \item If $\Gamma \ts S' \sub S$
    and $\Gamma,X<:S,\Delta \ts C_1 \sub C_2$,
    then $\Gamma,X<:S',\Delta \ts C_1 \sub C_2$.

  \item If $\Gamma \ts S' \sub S$
    and $\Gamma,X<:S,\Delta \ts U \sub T$,
    then $\Gamma,X<:S',\Delta \ts U \sub T$.

  \item If $\Gamma \ts S' \sub S$
    and $\Gamma,X<:S,\Delta \ts t \typ T$,
    then $\Gamma,X<:S',\Delta \ts t \typ T$.
  \end{enumerate}
\end{lemma}

\begin{proof}
  By straightforward induction on the derivations.
  Only subtyping considers types to which type variables are bound,
  and the only rule to do so is \ruleref{tvar}, which we prove below.
  All other cases follow from IH or other narrowing lemmas.

  \begin{proofcases}
    \pcase{\ruleref{tvar}} We need to prove $\Gamma,X<:S',\Delta \ts X \sub S$,
    which follows from weakening the lemma premise and using \ruleref{trans}
    together with \ruleref{tvar}.
  \end{proofcases}
\end{proof}

\begin{lemma}[Term binding narrowing]\ \\
  \begin{enumerate}
  \item If $\Gamma \ts U' \sub U$
    and $\Gamma,x:U,\Delta \ts T \wf$,
    then $\Gamma,x:U',\Delta \ts T \wf$.

  \item If $\Gamma \ts U' \sub U$
    and $\Gamma,x:U,\Delta \ts C_1 \sub C_2$,
    then $\Gamma,x:U',\Delta \ts C_1 \sub C_2$.

  \item If $\Gamma \ts U' \sub U$
    and $\Gamma,x:U,\Delta \ts T_1 \sub T_2$,
    then $\Gamma,x:U',\Delta \ts T_1 \sub T_2$.

  \item If $\Gamma \ts U' \sub U$
    and $\Gamma,x:U,\Delta \ts t \typ T$,
    then $\Gamma,x:U',\Delta \ts t \typ T$.
  \end{enumerate}
\end{lemma}

\begin{proof}
  By straightforward induction on the derivations.
  Only subcapturing and typing consider types to which term variables are bound.
  Only \ruleref{sc-var} and \ruleref{var} do so, which we prove below.
  All other cases follow from IH or other narrowing lemmas.

  \begin{proofcases}
    \pcase{\ruleref{var}}
    We know that $U = \CS{C}R$
    and $\Gamma,x:\CS{C}R,\Delta \ts x \typ \CAPT{x}R$.
    As $\Gamma \ts U' \sub U$, from \lref{lm:sub-inv-capt}
    we know that $U' = \CS{C'}R'$
    and that $\Gamma \ts R' \sub R$.
    We need to prove that
    $\Gamma,x:\CS{C'}R',\Delta \ts x \typ \CAPT{x}R$.
    We can do so through \ruleref{var}, \ruleref{sub},
    \ruleref{capt}, \ruleref{sc-elem}
    and weakening $\Gamma \ts R' \sub R$.

    \pcase{\ruleref{sc-var}}
    Then we know that $C_1 = \{y\}$
    and that $y:T \in \Gamma,x:U,\Delta$
    and that $\Gamma,x:U,\Delta \ts \cv(T) \sub C_2$.

    If $y \neq x$, we can conclude by IH and \ruleref{sc-var}.

    Otherwise, we have $T = U$.
    From \lref{lm:sub-inv-capt} we know that
    $\Gamma \ts \cv(U') \sub \cv(U)$,
    and from IH we know that
    $\Gamma,x:U',\Delta \ts \cv(U) \sub C_2$.
    By \ruleref{sc-var} to conclude it is enough to have
    $\Gamma,x:U',\Delta \ts \cv(U') \sub C_2$,
    which we do have by connecting two previous conclusions
    by weakening and \lref{lm:sc-trans}.
  \end{proofcases}
\end{proof}

\subsection{Substitution}

\subsubsection{Term Substitution}

\begin{lemma}[Term substitution preserves subcapturing]
  \label{lm:tm-subst-sc}
  If $\Gamma,x:U,\Delta \ts C_1 \sub C_2$
  and $\Gamma \ts D \sub cv(U)$,
  then $\Gamma,[x \mapsto D]\Delta \ts [x \mapsto D]C_1 \sub [x \mapsto D]C_2$.
\end{lemma}

\begin{proof}
  Define $\theta \triangleq [x \mapsto D]$. By induction on the subcapturing derivation.

  \begin{proofcases}
    \pcase{\ruleref{sc-elem}}
    Then $C_1 = \{y\}$ and $y \in C_2$.
    Inspect if $y = x$.
    If no, then our goal is $\Gamma,\theta\Delta \ts \{y\} \sub \theta{}C_2$.
    In this case, $y \in \theta{}C_2$,
    which lets us conclude by \ruleref{sc-elem}.
    Otherwise, we have $\theta{}C_2 = (C_2 \setminus \{x\}) \cup D$,
    as $x \in C_2$.
    Then our goal is
    $\Gamma,\theta{}\Delta \ts D \sub (C_2 \setminus \{x\}) \cup D$,
    which can be shown by \ruleref{sc-set} and \ruleref{sc-elem}.

    \pcase{\ruleref{sc-var}}
    Then $C_1 = \{y\}$
    and $y : C_3\;P' \in \Gamma,x:U,\Delta$
    and $\Gamma,x:U,\Delta \ts C_3 \sub C_2$.\\
    \indent Inspect if $y = x$.
    If yes, then our goal is
    $\Gamma,\theta\Delta \ts D \sub \theta{}C_2$.
    By IH we know that
    $\Gamma,\theta\Delta \ts \theta{}C_3 \sub \theta{}C_2$.
    As $x = y$, we have
    $U = C_3\,P'$
    and therefore based on an initial premise of the lemma we have
    $\Gamma \ts D \sub C_3$.
    Then by weakening and IH, we know that
    $\Gamma,\theta\Delta \ts \theta{}D \sub \theta{}C_3$,
    which means we can conclude by \lref{lm:sc-trans}.\\
    \indent Otherwise, $x \neq y$, and our goal is
    $\Gamma,\theta{}\Delta \ts C_1 \sub \theta{}C_2$.
    We inspect where $y$ is bound.
    \begin{proofcases}
      \pcase{$y \in \dom(\Gamma)$}
      Then $y \not\in C_3$, as $\Gamma,x:U,\Delta \wf$.
      By IH we have
      $\Gamma,\theta{}\Delta \ts \theta{}C_3 \sub \theta{}C_2$.
      We can conclude by \ruleref{sc-var}
      as $[x \mapsto D]C_3 = C_3$
      and $y: C_3\;P \in \Gamma,\theta{}\Delta$.

      \pcase{$y \in \dom(\Delta)$}
      Then $y: \theta{}(C_3\;P) \in \Gamma,\theta{}\Delta$
      and we can conclude by IH and \ruleref{sc-var}.
    \end{proofcases}

    \pcase{\ruleref{sc-set}}
    Then $C_1 = \{y_1, \ldots, y_n\}$
    and we inspect if $x \in C_1$.\\
    \indent If not, then for all $y \in C_1$
    we have $\theta{}\{y\} = \{y\}$
    and so we can conclude by repeated IH on our premises and \ruleref{sc-set}.\\
    \indent If yes, then we know that:
    $\forall\,y \in C_1.\, \Gamma,x:U,\Delta \ts \{y\} \sub C_2$.
    We need to show that
    $\Gamma,\theta{}\Delta \ts \theta{}C_1 \sub \theta{}C_2$.
    By \ruleref{sc-set}, it is enough to show that if $y' \in \theta{}C_1$,
    then $\Gamma,\theta{}\Delta \ts \{y'\} \sub \theta{}C_2$.
    For each such $y'$, there exists $y \in C_1$ such that $y' \in \theta{}\{y\}$.
    For this $y$, from a premise of \ruleref{sc-set} we know that
    $\Gamma,x:U,\Delta \ts \{y\} \sub \theta{}C_2$
    and so by IH we have
    $\Gamma,\theta{}\Delta \ts \theta{}\{y\} \sub \theta{}C_2$.
    Based on that, by \lref{lm:sc-trans} we also have
    $\Gamma,\theta{}\Delta \ts \{y'\} \sub \theta{}C_2$.
    which is our goal.
  \end{proofcases}
\end{proof}

\begin{lemma}[Term substitution preserves subtyping]
  \label{lm:tm-subst-sub}
  If $\Gamma,x:U,\Delta \ts U \sub T$
  and $\Gamma \ts y:U$,
  then $\Gamma,[x \mapsto y]\Delta \ts [x \mapsto y]U \sub [x \mapsto y]T$.
\end{lemma}

\begin{proof}
  Define $\theta \triangleq [x \mapsto y]$.
  Proceed by induction on the subtyping derivation.

  \begin{proofcases}
    \pcase{\ruleref{refl}, \ruleref{top}} By same rule.

    \pcase{\ruleref{capt}} By IH and \lref{lm:tp-subst-sc} and \ruleref{capt}.

    \pcase{\ruleref{trans}, \ruleref{boxed}, \ruleref{fun}, \ruleref{tfun}}
    By IH and re-application of the same rule.

    \pcase{\ruleref{tvar}}
    Then $U = Y$ and $T = C\;R$
    and $Y <: C\;R \in \Gamma,x:U,\Delta$
    and our goal is $\Gamma,x:U,\Delta \ts \theta{}Y \sub \theta{}(C\;R)$.
    Note that $x \neq Y$ and inspect where $Y$ is bound.
    If $Y \in \dom(\Gamma)$,
    we have $Y <: C\;R \in \Gamma,\theta\Delta$
    and since $X \not\in \fv(R)$ (as $\Gamma,X<:R,\Delta \wf$),
    $\theta(C\;R) = C\;R$.
    Then, we can conclude by \ruleref{tvar}.
    Otherwise if $Y \in \dom(\Delta)$,
    we have $Y <: C\;\theta{}R \in \Gamma,\theta\Delta$
    and we can conclude by \ruleref{tvar}.
  \end{proofcases}
\end{proof}

\begin{lemma}[Term substitution preserves typing]
  \label{lm:tm-subst-typ}
  If $\Gamma,x:U,\Delta \ts t \typ T$
  and $\Gamma \ts y \typ U$,
  then $\Gamma,[x \mapsto y]\Delta \ts [x \mapsto y]t \typ [x \mapsto y]T$.
\end{lemma}
\begin{proof}
  Define $\theta \triangleq [x \mapsto y]$.
  Proceed by induction on the typing judgement.

  \begin{proofcases}
    \pcase{\ruleref{var}}
    Then $t = z$
    and $z:C\,R \in \Gamma,x:U,\Delta$
    and $T = \{z\}\,R$
    and our goal is $\Gamma,\theta\Delta \ts z \typ \theta(\{z\}\,R)$.\\
    \indent If $z = x$,
    then $U = C\,R$
    and $\theta(\{x\}\,R) = \{y\}\,R$
    and by \ruleref{var} $\Gamma,\theta\Delta \ts y \typ \{y\}\,R$.\\
    \indent Otherwise, $z \neq x$ and we inspect $z$ it is bound.\\
    \indent\indent If $z \in \dom(\Gamma)$,
    then $x \not\in \fv(C\,R)$
    and so $\theta(\{z\}\,R) = \{z\}\,R$
    and we can conclude by \ruleref{var}.\\
    \indent\indent Otherwise, $z \in \dom(\Delta)$
    and $z:\theta(C\,R) \in \Gamma,\theta\Delta$
    and we can conclude by \ruleref{var}.

    \pcase{\ruleref{sub}}
    By IH, \lref{lm:tm-subst-sub} and \ruleref{sub}.

    \pcase{\ruleref{abs}}
    Then $\Gamma,x:U,\Delta,x':U' \ts t \typ T$.\\
    \indent By IH, we have that $\Gamma,\theta\Delta, x':\theta{}U' \ts t \typ \theta{}T$.
    If we note that term substitution does not affect $\cv$,
    the result immediately follows.

    \pcase{\ruleref{tabs}} Similar to previous rule.

    \pcase{\ruleref{app}}
    Then $t = z\,z'$
    and $\Gamma,x:U,\Delta \ts z \typ C\,\ARR{x' : U'}{T'}$
    and $\Gamma,x:U,\Delta \ts z' \typ U'$
    and $T = [x' \mapsto z']T'$.\\
    \indent By IH we have
    $\Gamma,\theta\Delta \ts \theta{}z \typ \theta(C\,\ARR{x' : U'}{T'})$
    and $\Gamma,\theta\Delta \ts \theta{}z' \typ \theta{}U'$.\\
    \indent Then by \ruleref{app} we have
    $\Gamma,\theta\Delta \ts \theta(z\,z') \typ [x' \mapsto \theta{}z']\theta{}T'$.\\
    \indent As $x' \neq x$, we have
    $[x' \mapsto \theta{}z']\theta{}T' = \theta([x' \mapsto z']T')$,
    which concludes.

    \pcase{\ruleref{tapp}} Similar to previous rule.

    \pcase{\ruleref{box}}
    Then $t = \Boxed{z}$
    and $\Gamma,x:U,\Delta \ts z \typ C\,R$
    and $T = \Boxed{C\,R}$.\\
    By IH, we have $\Gamma,\theta\Delta \ts \theta{}z \typ \theta{}C\,\theta{}R$.
    If $x \not\in C$,
    we have $\theta{}C = C$
    and $C \subseteq \dom(\Gamma,\theta\Delta)$
    which lets us conclude by \ruleref{box}.
    Otherwise, $\theta{}C = (C \setminus \{x\}) \cup \{y\}$
    As $\Gamma \ts y \typ U$,
    $\theta{}C \subseteq \dom(\Gamma,\theta\Delta)$,
    which again lets us conclude by \ruleref{box}.

    \pcase{\ruleref{unbox}}
    Analogous to the previous rule.
    Note that we just swap the types in the premise and the conclusion.
  \end{proofcases}
\end{proof}

\subsubsection{Type Substitution}

\begin{lemma}[Type substitution preserves subcapturing]
  \label{lm:tp-subst-sc}
  If $\Gamma,X<:R,\Delta \ts C \sub D$
  and $\Gamma \ts P <: R$
  then $\Gamma,[X \mapsto P]\Delta \ts C \sub D$.
\end{lemma}

\begin{proof}
  Define $\theta \triangleq [X \mapsto R]$.
  Proceed by induction on the subcapturing judgement.
  \begin{proofcases}
    \pcase{\ruleref{sc-set}, \ruleref{sc-elem}} By IH and same rule.

    \pcase{\ruleref{sc-var}}
    Then $C = \{y\}$, $y : C'\;R' \in \Gamma,X<:R,\Delta$, $y \neq X$.
    Inspect where $y$ is bound.
    If $y \in \dom(\Gamma)$,
    we have $y : C'\;R' \in \Gamma,\theta\Delta$.
    Otherwise, by definition of substition
    we have $y : C'\;\theta{}R' \in \Gamma,\theta\Delta$ .
    In both cases we can conclude by \ruleref{sc-var}.
  \end{proofcases}
\end{proof}

\begin{lemma}[Type substitution preserves subtyping]
  \label{lm:tp-subst-sub}
  If $\Gamma,X<:R,\Delta \ts U \sub T$
  and $\Gamma \ts P \sub R$,
  then $\Gamma,[X \mapsto P]\Delta \ts [X \mapsto P]U \sub [X \mapsto P]T$.
\end{lemma}

\begin{proof}
  Define $\theta \triangleq [X \mapsto P]$.
  Proceed by induction on the subtyping derivation.

  \begin{proofcases}
    \pcase{\ruleref{refl}, \ruleref{top}} By same rule.

    \pcase{\ruleref{capt}} By IH and \lref{lm:tp-subst-sc} and \ruleref{capt}.

    \pcase{\ruleref{trans}, \ruleref{boxed}, \ruleref{fun}, \ruleref{tfun}}
    By IH and re-application of the same rule.

    \pcase{\ruleref{tvar}}
    Then $U = Y$ and $T = C\;R'$
    and $Y <: C\;R' \in \Gamma,X<:R,\Delta$
    and our goal is $\Gamma,X<:R,\Delta \ts \theta{}Y \sub \theta{}(C\;R')$.
    If $Y = X$, by lemma premise and weakening.
    Otherwise, inspect where $Y$ is bound.
    If $Y \in \dom(\Gamma)$,
    we have $Y <: C\;R' \in \Gamma,\theta\Delta$
    and since $X \not\in \fv(R')$ (as $\Gamma,X<:R,\Delta \wf$),
    $\theta(C\;R') = C\;R'$.
    Then, we can conclude by \ruleref{tvar}.
    Otherwise if $Y \in \dom(\Delta)$,
    we have $Y <: C\;\theta{}R' \in \Gamma,\theta\Delta$
    and we can conclude by \ruleref{tvar}.
  \end{proofcases}
\end{proof}

\begin{lemma}[Type substitution preserves typing]
  \label{lm:tp-subst-typ}
  If $\Gamma,X<:R,\Delta \ts t \typ T$
  and $\Gamma \ts P \sub R$,
  then $\Gamma,[X \mapsto P]\Delta \ts [X \mapsto P]t \typ [X \mapsto P]T$.
\end{lemma}

\begin{proof}
  Define $\theta \triangleq [X \mapsto P]$.
  Proceed by induction on the typing derivation.
  \begin{proofcases}
    \pcase{\ruleref{var}}
    Then $t = y$, $y : C'\;R' \in \Gamma,X<:R,\Delta$, $y \neq X$,
    and our goal is $\Gamma,\theta\Delta \ts x : \{x\} \theta{}R'$.\\
    \indent Inspect where $y$ is bound.
    If $y \in \dom(\Gamma)$, {
      then $y : C'\;R' \in \Gamma,\theta\Delta$
      and $X \not\in \fv(R')$ {
        (as $\Gamma,X<:R,\Delta \wf$) }.
      Then, $\theta(C'\;R') = C'\;R'$
      and we can conclude by \ruleref{var}.
    }.
    Otherwise, $y : C'\;\theta{}R' \in \Gamma,\theta\Delta$
    and we can directly conclude by \ruleref{var}.

    \pcase{\ruleref{abs}, \ruleref{tabs}}
    In both rules, observe that type substitution does not affect $\cv$
    and conclude by IH and rule re-application.

    \pcase{\ruleref{app}}
    Then we have $t = x\;y$
    and $\Gamma,X<:R,\Delta \ts x : C\;\ARR{z : U}T_0$
    and $T = [z \mapsto y]T_0$.\\
    \indent We observe that $\theta[z \mapsto y]T_0 = [z \mapsto y]\theta{}T_0$
    and $\theta{}t = t$
    and conclude by IH and \ruleref{app}.

    \pcase{\ruleref{tapp}}
    Then we have $t = x\,[R']$
    and $\Gamma,X<:R,\Delta \ts x : C\;\TARR{Z <: R'}T_0$
    and $T = [Z \mapsto R']T_0$.\\
    \indent We observe that
    $\theta[Z \mapsto R']T_0 = [Z \mapsto \theta{}R']\theta{}T_0$.
    By IH, $\Gamma,\theta\Delta \ts x : C\;\TARR{Z : \theta{}R'}T_0$,
    Then, we can conclude by \ruleref{tapp}.

    \pcase{\ruleref{box}}
    Then $t = \Boxed{y}$
    and $\Gamma,X<:R,\Delta \ts y \typ C\;R$
    and $T = \Boxed{C\;R}$,
    and our goal is $\Gamma,\theta\Delta \ts y : \Boxed{\theta(C'\;R')}$.\\
    \indent Inspect where $y$ is bound.
    If $y \in \dom(\Gamma)$, {
      then $y : C'\;R' \in \Gamma,\theta\Delta$
      and $X \not\in \fv(R')$ {
        (as $\Gamma,X<:R,\Delta \wf$) }.
      Then, $\theta(C'\;R') = C'\;R'$
      and we can conclude by \ruleref{box}.
    }
    Otherwise, $y : C'\;\theta{}R' \in \Gamma,\theta\Delta$
    and we can directly conclude by \ruleref{box}.

    \pcase{\ruleref{unbox}}
    Proceed analogously to the case for \ruleref{box} -- we just swap
    the types in the premise and in the consequence.

    \pcase{\ruleref{sub}} By IH and \ref{lm:tp-subst-sub}.
  \end{proofcases}
\end{proof}

\subsection{Main Theorems -- Soundness}

\subsubsection{Preliminaries}

As we state Preservation (Theorem \ref{th:preservation}) in a non-empty
environment, we need to show canonical forms lemmas in such an environment as
well. To do so, we need to know that values cannot be typed with a type that is
a type variable, which normally follows from the environment being empty.
Instead, we show the following lemma:

\begin{lemma}[Value typing]
  \label{lm:val-typ}
  If $\Gamma \ts v \typ T$,
  then $T$ is not of the form $\CS{C}X$.
\end{lemma}

\begin{proof}
  By induction on the typing derivation.\\
  \indent For rule \ruleref{sub},
  we know that $\Gamma \ts v \typ U$ and $\Gamma \ts U \sub T$.
  Assuming $T = \CS{C}X$,
  we have a contradiction
  by \lref{lm:sub-inv-tvar} and IH.\\
  \indent Rules \ruleref{box}, \ruleref{abs}, \ruleref{tabs} are immediate,
  and other rules are not possible.
\end{proof}

\begin{lemma}[Canonical forms: term abstraction]
  \label{lm:canon-abs}
  If $\Gamma \ts v \typ \CS{C}\ARR{x : U}T$,
  then $v = \LAM{x : U'}t$
  and $\Gamma \ts U \sub U'$
  and $\Gamma,x:U \ts t \typ T$.
\end{lemma}

\begin{proof}
  By induction on the typing derivation.

  For rule \ruleref{sub},
  we observe that by \lref{lm:sub-inv-fun} and by \lref{lm:val-typ},
  the subtype is of the form $\CS{C'}\ARR{y:U''}T'$
  and we have $\Gamma \ts U \sub U''$.
  By IH we know that $v = \LAM{x:U'}t$
  and $\Gamma \ts U'' \sub U'$
  and $\Gamma,x:U'' \ts t \typ T$.
  By \ruleref{trans} we have $\Gamma \ts U \sub U'$
  and by narrowing we have $\Gamma,x:U \ts t \typ T$,
  which concludes.

  Rule \ruleref{abs} is immediate, and other rules cannot occur.
\end{proof}

\begin{lemma}[Canonical forms: type abstraction]
  \label{lm:canon-tabs}
  If $\Gamma \ts v \typ \CS{C}\TARR{X : U}T$,
  then $v = \TLAM{X : U'}t$
  and $\Gamma \ts U \sub U'$
  and $\Gamma,X<:U \ts t \typ T$.
\end{lemma}

\begin{proof}
  Analogous to the proof of \lref{lm:canon-abs}.
\end{proof}

\begin{lemma}[Canonical forms: boxed term]
  \label{lm:canon-box}
  If $\Gamma \ts v \typ \CS{C}\Boxed{T}$,
  then $v = \Boxed{x}$
  and $\Gamma \ts x \typ T$.
\end{lemma}

\begin{proof}
  Analogous to the proof of \lref{lm:canon-abs}.
\end{proof}

\begin{lemma}[Variable lookup inversion]
  \label{lm:inv-var-lookup}
  If we have both $\Gamma \ts \sigma \sim \Delta$ and $x : C\,R \in \Gamma,\Delta$,
  then $\sigma(x) = v$ implies that $\Gamma,\Delta \ts v \typ \CS{C}R$.
\end{lemma}

\begin{proof}
  By structural induction on $\sigma$.
  It is not possible for $\sigma$ to be empty.\\
  \indent Otherwise, $\sigma = \sigma'[\Let{y}{v}{[\,]}]$
  and for some $U$ we have both
  $\Delta = \Delta',y:U$
  and $\Gamma,\Delta' \ts v \typ U$.
  \note{We remove bindings from the ``inside'' in order to make induction
    possible below.}
  \\
  \indent If $y \neq x$,
  we can proceed by IH as $x$ can also be typed in $\Gamma,\Delta'$,
  after which we can conclude by weakening.
  Otherwise, we proceed by induction on the typing derivation.
  \begin{proofcases}
    \pcase{\ruleref{var}}
    Then $C = \{x\}$.
    Using subsumption and \ruleref{sc-var},
    we can conclude with $U = \CS{C'}R$.

    \pcase{\ruleref{sub}}
    Then we have $\Gamma,\Delta \ts x \typ \CS{C''}R'$
    and $\Gamma,\Delta \ts \CS{C''}R' \sub \CS{C}R$.
    By \lref{lm:sub-inv-capt} we have both
    $\Gamma,\Delta \ts C'' \sub C$
    and $\Gamma,\Delta \ts R' \sub R$.
    Then we proceed by using the IH on $\Gamma,\Delta \ts x \typ \CS{C''}R'$
    and conclude through subsumption and transitivity (\ref{fact:trans}).

    \noindent Other rules are impossible.
  \end{proofcases}
\end{proof}

\begin{lemma}[Term abstraction lookup inversion]
  \label{lm:inv-tm-abs}
  If $\Gamma \ts \sigma \sim \Delta$
  and $\Gamma,\Delta \ts x \typ \CS{C}\ARR{z: U}T$
  and $\sigma(x) = \LAM{z: U'}t$,
  then $\Gamma,\Delta \ts U \sub U'$
  and $\Gamma,\Delta,z:U \ts t \typ T$.
\end{lemma}

\begin{proof}
  A corollary of \lref{lm:inv-var-lookup} and \lref{lm:canon-abs}.
\end{proof}

\begin{lemma}[Type abstraction lookup inversion]
  \label{lm:inv-tp-abs}
  If $\Gamma \ts \sigma \sim \Delta$
  and $\Gamma,\Delta \ts x \typ \CS{C}\Ttlam{Z}{U}{T}$
  and $\sigma(x) = \tlam{Z}{U'}{t}$,
  then $\Gamma,\Delta \ts U \sub U'$
  and $\Gamma,\Delta,Z<:U \ts t \typ T$.
\end{lemma}

\begin{proof}
  A corollary of \lref{lm:inv-var-lookup} and \lref{lm:canon-tabs}.
\end{proof}

\begin{lemma}[Box lookup inversion]
  \label{lm:inv-box}
  If $\Gamma \ts \sigma \sim \Delta$
  and $\sigma(x) = \Boxed{y}$
  and $\Gamma,\Delta \ts x \typ \Boxed{T}$,
  then $\Gamma,\Delta \ts y \typ T$.
\end{lemma}

\begin{proof}
  A corollary of \lref{lm:inv-var-lookup} and \lref{lm:canon-box}.
\end{proof}

\subsubsection{Soundness}

In this section, we show the classical soundness theorems.

\begin{theorem}[Preservation]
  \label{th:preservation}
  If we have $\Gamma \ts \sigma \sim \Delta$ and $\Gamma,\Delta \ts t : T$,
  then $\storeof{t} \reduces \storeof{t'}$
  implies that $\Gamma,\Delta \ts t' : T$.
\end{theorem}

\begin{proof}
  We proceed by inspecting the rule used to reduce $\storeof{\evalof{t}}$.

  \begin{proofcases}
    \pcase{\eruleref{apply}}
    Then we have $t = \evalof{x\,y}$
    and $\sigma(x) = \LAM{z: U}s$
    and $t' = \evalof{[z \mapsto y]s}$.

    By Lemma~\ref{lm:ectx-plug-typ}, $x\,y$ is well-typed in $\Gamma,\Delta$
    and we proceed by induction on its typing derivation.
    Only two rules could apply:
    \ruleref{sub} and \ruleref{app}.
    In the first case, we can conclude by IH and subsumption.
    In the second case, we have:
    $\Gamma,\Delta \ts x \typ \ARR{z: U_0}T_0$
    and $\Gamma,\Delta \ts y \typ U_0$.
    By \lref{lm:inv-tm-abs} and narrowing, we have
    $\Gamma,\Delta,z:U_0 \ts t_0 : T_0$.
    By \lref{lm:tm-subst-typ}, we can type the substituted $s$ as
    $\Gamma,\Delta \ts [z \mapsto y]s : [z \mapsto y]T'_0$.
    We can now conclude by restoring the contexts $\sigma$ and $e$
    with Lemmas \ref{lm:ectx-replug} and \ref{lm:tm-rectx-store}.

    \pcase{\eruleref{tapply}, \eruleref{unbox}} As above.

    \pcase{\eruleref{rename}}
    Then we have
    $t = \evalof{\Let{x}{y}{t_0}}$
    and $t' = \evalof{[x \mapsto y]t_0}$.

    As $\Gamma,\Delta \ts T_0 \wf$ and by Barendregt $x \not\in \dom \Gamma,\Delta$, we have $x \not\in \fv(T_0)$.
    We again proceed by induction on the typing derivation of the plug inside $t$, that is $\Let{x}{y}{t_0}$.
    Only two rules could apply:
    \ruleref{sub} and \ruleref{let}.
    In the first case, we can conclude by IH and subsumption.
    In the second case for some $U$ we have
    $\Gamma,\Delta,x:U \ts t_0 \typ T_0$
    and $\Gamma,\Delta \ts y \typ U$.
    Then by \lref{lm:tm-subst-typ},
    $\Gamma,\Delta \ts [x \mapsto y]t_0 \typ [x \mapsto y]T_0$.
    Since $x \not\in \fv(T_0)$
    we have $[x \mapsto y]T_0 = T_0$,
    and so by \lref{lm:ectx-replug} we also have
    $\Gamma,\Delta \ts \evalof{[x \mapsto y]t_0} \typ T$,
    which by \lref{lm:tm-rectx-store} concludes.

    \pcase{\eruleref{lift}}
    Then we have $t = \evalof{\Let{x: U}{v}{t_0}}$
    and $t' = \Let{x}{v}{\evalof{t_0}}$.

    Like in the previous case,
    we proceed by induction on the typing derivation of the plug inside $t$.
    The \ruleref{sub} is, again, trivial.
    The only other case is again \ruleref{let},
    where for some $U$ we have
    $\Gamma,\Delta,x:U \ts t_0 \typ T_0$.
    By \lref{lm:sctx-plug-typ}, we have
    $\Gamma,\Delta \ts \evalof{t} \typ T$.
    As $t$ and $t_0$ have same types in $\Gamma,\Delta$,
    we can replug $e$ (\lref{lm:ectx-replug})
    and weaken the typing context to obtain
    $\Gamma,\Delta,x:U \ts \evalof{t_0} \typ T$.
    From the Barendregt convention, we know that $x \not\in \fv(T)$
    and so by \ruleref{let} we have
    $\Gamma,\Delta \ts \Let{x}{v}{\evalof{t_0}} \typ T$.
    To conclude, we restore $\sigma$ with \lref{lm:tm-rectx-store}.
  \end{proofcases}
\end{proof}

\begin{definition}[Canonical store-plug split]
  We say that a term of the form $\storeof{t}$ is a {\it canonical split} (of
  the entire term into store context $\sigma$ and the plug $t$) if $t$ is not of
  the form $\Let{x: T}{v}{t'}$.
\end{definition}

\begin{fact}
  Every term has a unique canonical store-plug split, and finding it is decidable.
\end{fact}

\begin{lemma}[Extraction of bound value]
  \label{lm:extract-val}
  If $\Gamma,\Delta \ts x \typ T$
  and $\ts \sigma \sim \Delta$
  and $x \in \dom(\Delta)$,
  then $\sigma(x) = v$.
\end{lemma}

\begin{proof}
  By structural induction on $\Delta$.
  If $\Delta$ is empty, we have a contradiction.
  Otherwise, $\Delta = \Delta',z:T'$
  and $\sigma = \sigma'[\Let{z:T'}{v}{[\,]}]$
  and $\Gamma,\Delta',z:T' \ts v \typ T'$.
  Note that $\Delta$ is the environment matching $\sigma$
  and can only contain term bindings.
  If $z = x$, we can conclude immediately,
  and otherwise if $z \neq x$, we can conclude by IH.
\end{proof}

\begin{theorem}[Progress]
  \label{th:progress}
  If $\ts \storeof{\evalof{t}} \typ T$
  and $\storeof{\evalof{t}}$ is a canonical store-plug split,
  then either $\evalof{t} = a$,
  or there exists $\storeof{t'}$ such that $\storeof{\evalof{t}} \reduces \storeof{t'}$.
\end{theorem}

\begin{proof}
  Since $\storeof{\evalof{t}}$ is well-typed in the empty environment, there clearly must be some $\Delta$ such that $\ts \sigma \sim \Delta$ and $\Delta \ts \evalof{t} : T$.
  By Lemma~\ref{lm:ectx-plug-typ}, we have that $\Delta \ts t \typ U$ for some $U$.
  We proceed by induction on the derivation of this judgement.

  \begin{proofcases}
    \pcase{\ruleref{var}}
    Then $t = x$.\\
    \indent If $e$ is non-empty, $\evalof{x} = e'[\ \Let{y: T}{x}{t}\ ]$
    and we can step by \eruleref{rename};
    otherwise, immediate.

    \pcase{\ruleref{abs}, \ruleref{tabs}, \ruleref{box}}
    Then $t = v$.\\
    \indent If $e$ is non-empty, $\evalof{v} = e'[\ \Let{x: T}{v}{t}\ ]$
    and we can step by \eruleref{lift};
    otherwise, immediate.

    \pcase{\ruleref{app}}
    Then $t = x\;y$
    and $x : \CS{C}\ARR{z : U_0}T_0 \in \Delta$
    and $\Delta \ts y \typ U_0$.\\
    \indent By Lemmas \ref{lm:extract-val} and \ref{lm:canon-abs},
    $\sigma(x) = \LAM{z : U_1}t'$,
    which means we can step by \eruleref{apply}.

    \pcase{\ruleref{tapp}}
    Then $t = x\;[R]$
    and $x : \CS{C}\TARR{Z <: R}T_0 \in \Delta$.\\
    \indent By Lemmas \ref{lm:extract-val} and \ref{lm:canon-tabs},
    $\sigma(x) = \tlam{z}{R}{t'}$,
    which means we can step by \eruleref{tapply}.

    \pcase{\ruleref{unbox}}
    Then $t = \Unboxed{C}{x}$
    and $x : \CS{C}\Boxed{C\;R} \in \Delta$.\\
    \indent By Lemmas \ref{lm:extract-val} and \ref{lm:canon-box},
    $\sigma(x) = \Boxed{y}$,
    which means we can step by \eruleref{open}.

    \pcase{\ruleref{let}}
    Then $t = \Let{x : T}{s}{t'}$
    and we proceed by IH on $s$,
    with $e[\ \Let{x: T}{[\,]}{t'}\ ]$ as the evaluation context.

    \pcase{\ruleref{sub}} By IH.
  \end{proofcases}
\end{proof}

\subsubsection{Consequences}

\begin{lemma}[Capture prediction for answers]
  \label{lm:answer-cs}
  If $\Gamma \ts \sigma[a] \typ \CS{C}R$,
  then $\Gamma \ts \sigma[a] \typ \CS{\cv(\sigma[a])}R$
  and $\Gamma \ts \cv(\sigma[a]) \sub C$.
\end{lemma}

\begin{proof}
  By induction on the typing derivation.

  \begin{proofcases}
    \pcase{\ruleref{sub}}
    Then $\Gamma \ts a \typ \CS{C'}R'$
    and $\Gamma \ts \CS{C'}R' \sub \CS{C}R$.
    By IH, $\Gamma \ts \sigma[a] \typ \CS{\cv(\sigma[a])}R'$
    and $\Gamma \ts \cv(\sigma[a]) \sub C'$.
    By \lref{lm:sub-inv-capt},
    we have that $\Gamma \ts C' \sub C$ and
    $\Gamma \ts R' \sub R$.

    To conclude we need
    $\Gamma \ts \sigma[a] \typ \CS{\cv(\sigma[a])}R$
    and $\Gamma \ts \cv(\sigma[a]) \sub C$,
    which we respectively have
    by subsumption and \lref{lm:sc-trans}.

    \pcase{\ruleref{var}, \ruleref{abs}, \ruleref{tabs}, \ruleref{box}}
    Then $\sigma$ is empty and $C = \cv(a)$.
    One goal is immediate,
    other follows from \lref{lm:sc-refl}.

    \pcase{\ruleref{let}}
    Then $\sigma = \Let{x}{v}{\sigma'}$
    and $\Gamma,x:U \ts \sigma'[a] \typ \CS{C}R$
    and $x \not\in C$.

    By IH,
    $\Gamma,x:U \ts \sigma'[a] \sub \CS{\cv(\sigma'[a])}R$
    and $\Gamma,x:U \ts \cv(\sigma'[a]) \sub C$.

    By \lref{lm:tm-subst-sc}, we have
    $\Gamma \ts [x \mapsto \cv(v)](\cv(\sigma'[a])) \sub [x \mapsto \cv(v)]C$.

    By definition,
    $[x \mapsto \cv(v)](\cv(\sigma'[a])) = \cv(\Let{x}{v}{\sigma'[a]})$,
    and we also already know that $x \not\in C$.

    This lets us conclude, as we have
    $\Gamma \ts \cv(\Let{x}{v}{\sigma'[a]}) \sub C$.

    \noindent Other rules cannot occur.
  \end{proofcases}
\end{proof}

\begin{lemma}[Capture prediction for terms]
  Let $\ts \sigma \sim \Delta$ and $\Delta \ts t \typ \CS{C}R$.\\
  \quad Then $\sigma[t] -->^{*} \sigma[\sigma'[a]]$
  implies that $\Delta \ts \cv(\sigma'[a]) \sub C$.
\end{lemma}

\begin{proof}
  By preservation, $\ts \sigma'[a] \typ \CS{C}R$, which lets us conclude by \lref{lm:answer-cs}.
\end{proof}

\subsection{Correctness of boxing}

\newcommand\Resolve{\mathrm{resolve}}
\newcommand\Resolver{\mathrm{resolver}}

\subsubsection{Relating cv and stores}
We want to relate the $\cv$ of a term of the form $\sigma[t]$ with $\cv(t)$.
We are after a definition like this:
$$
  \cv(\sigma[t]) = \Resolve(\sigma, \cv(t))
$$
\noindent
Let us consider term of the form $\sigma[t]$ and a store of the form $\kw{let} x = v \kw{in} \sigma'$.
There are two rules that could be used to calculate $\cv(\kw{let} x = v \kw{in} \sigma')$:
\begin{align*}
  \cv(\Let{x}{v}{t}) &= \cv(t) \gap\gap\gap \kw{if} x \notin \cv(t) \\
  \cv(\Let{x}{s}{t}) &= \cv(s) \cup \cv(t) \backslash x
\end{align*}
\noindent
Observe that since we know that $x$ is bound to a value, we can reformulate these rules as:
\begin{align*}
  \cv(\Let{x}{v}{t}) = [x |-> \cv(v)]\cv(t)
\end{align*}
\noindent
Which means that we should be able to define $\Resolve$ with a substitution.
We will call this substitition a \emph{store resolver}, and we define it as:
\begin{align*}
  \Resolver(\kw{let} x = v \kw{in} \sigma) &= [x |-> \cv(v)] \circ \Resolver(\sigma) \\
  \Resolver([]) &= id
\end{align*}
\noindent
Importantly, note that we use \emph{composition} of substitutions. We have:
$$
\Resolve(\kw{let} x = a \kw{in} \kw{let} y= x \kw{in} []) \equiv [x |-> \{a\}, y |-> \{a\}]
$$

\noindent
With the above, we define $\Resolve$ as:
$$
  \Resolve(\sigma, C) = \Resolver(\sigma)(C)
$$
This definition satisfies our original desired equality with $\cv$.

\subsubsection{Relating cv and evaluation contexts}
Now to connect with evaluation contexts. We can extend $\cv$ to evaluation contexts like this:
\begin{align*}
  \cv(\kw{let} x = e \kw{in} t) &= \cv(e) \cup \cv(t)\\
  \cv([]) &= \{\}
\end{align*}

Now, the problem is that we can have ``degenerate'' evaluation contexts syntactically.
For instance we can split a term like $\kw{let} x = v \kw{in} z$ into a form like $(\kw{let} x = [] \kw{in} z)[v]$, even though this term is in normal form.
The problem with such a term is that $\cv$ may or may not count $\cv(v)$ as the $\cv$ of the overall term, which makes the relationship too complicated.
Instead, we will depend on reduction, like this:

\begin{fact}
  Let $\storeof{\evalof{t}}$ be a term that can reduce.
  Then: $$\cv(\storeof{\evalof{t}}) = \mathrm{resolve}(\sigma, \cv(e) \cup \cv(t))$$
\end{fact}

\subsubsection{Correctness of cv}

\begin{definition}[Platform environment]\label{def:platform-env}\forcenl
  \indent
  $\Gamma$ is a platform environment if for all $x \in \dom \Gamma$ we have $x : \UC\,S \in \Gamma$ for some $S$.
\end{definition}

\begin{lemma}[Inversion of subcapturing under platform environment]\label{lm:platform-sc-inversion}\forcenl
  \indent
  If $\Gamma$ is a platform environment and $\Gamma |- C <: D$, then either $C \subseteq D$ or $\rcap \in D$.
\end{lemma}
\begin{proof}
  By induction on the subcapturing relation.
  Case \ruleref{sc-elem} trivially holds.
  Case \ruleref{sc-set} holds by repeated IH.
  In case \ruleref{sc-var}, we have $C = \{x\}$ and $x : C'\,S \in \Gamma$.
  Since $\Gamma$ is a platform environment, we have $C' = \UC$, which means that the other premise of \ruleref{sc-var} is $\Gamma |- \UC <: D$.
  Since $\Gamma$ is well-formed, $\rcap \not\in \dom \Gamma$, which means that we must have $\rcap \in D$.
\end{proof}

\begin{lemma}[Strengthening of subcapturing]
  If $\Gamma,\Gamma' |- C <: D$ and $C \subseteq \dom \Gamma$, then we must have $\Gamma |- C <: D$.
\end{lemma}
\begin{proof}
  First, we consider that if $\rcap \in D$, we trivially have the desired goal.
  If $\rcap \not\in D$, we proceed by induction on the subcapturing relation.
  Case \ruleref{sc-elem} trivially holds and case \ruleref{sc-set} holds by repeated IH.

  In case \ruleref{sc-var}, we have $C = \{x\}$, $x : C'\,S \in \Gamma,\Gamma'$.
  This implies that $\Gamma = \Gamma_1,x:C'\,S,\Gamma_2$ (as $x \not\in \dom \Gamma$).
  Since $\Gamma,\Gamma'$ is well-formed, we must have $\Gamma_1 |- C' \wf$.
  Since we already know $\rcap \not\in D$, then we must also have $\rcap \not\in C'$,
  which then leads to $C' \subseteq \dom \Gamma_1$.
  This in turn means that by IH and weakening we have $\Gamma |- C' <: D$,
  and since we also have $x : C'\,S \in \Gamma$, we can conclude by \ruleref{sc-var}.
\end{proof}

Then we will need to connect it to subcapturing, because the keys used to open boxes are supercaptures of the capability inside the box. We want:

\begin{lemma}
  \label{lm:resolve-sc-relation}
  Let $\Gamma$ be a platform environment, $\Gamma |- \sigma \sim \Delta$ and $\Gamma,\Delta |- C_1 <: C_2$.
  Then $\Resolve(\sigma, C_1) \subseteq \Resolve(\sigma, C_2)$.
\end{lemma}
\begin{proof}
  By induction on $\sigma$.
  If $\sigma$ is empty, we have $\Resolve(\sigma, C_1) = C_1$, likewise for $C_2$, and we can conclude by Lemma~\ref{lm:platform-sc-inversion}.

  Otherwise, $\sigma = \sigma'[\kw{let} x = v \kw{in}]$ and $\Delta = \Delta',x : D_x\,S_x$ for some $S_x$.
  Let $\theta = \Resolver(\sigma)$.
  We proceed by induction on the subcapturing derivation.
  Case \ruleref{sc-elem} trivially holds and case \ruleref{sc-set} holds by repeated IH.

  In case \ruleref{sc-var}, we have $C_1 = \{y\}$ and $y : D_y\,S_y \in \Gamma,\Delta$ for some $S_y$, and $\Gamma,\Delta |- D_y <: C_2$.
  We must have $\Gamma,\Delta' |- D_y \wf$ and so we can strengthen subcapturing to $\Gamma,\Delta' |- D_y <: C_2$,
  which by IH gives us $\Resolver(\sigma')(D_y) \subseteq \Resolver(\sigma')(C_2)$.
  By definition we have $\theta = \Resolver(\sigma) = \Resolver(\sigma') \circ [x |-> \cv(v)]$.
  Since by well-formedness $x \not\in D_y$, we now have:
  $$
  \theta D_y \subseteq \theta C_2
  $$
  \noindent
  By Lemma~\ref{lm:inv-var-lookup} and Lemma~\ref{lm:answer-cs}, we must have $\Gamma,\Delta |- \cv(v) <: D_y$.
  Since $\Gamma,\Delta |- \cv(v) \wf$, we can strengthen this to $\Gamma,\Delta |- \cv(v) <: D_y$.
  By outer IH this gives us $\Resolver(\sigma')(\cv(v)) \subseteq \Resolver(\sigma')(D_y)$.
  Since $x \notin \cv(v) \cup D_y$, we have:
  $$
  \theta\cv(v) \subseteq \theta D_y
  $$
  \noindent
  Which means we have $\theta\cv(v) \subseteq \theta C_2$ and we can conclude by $\theta\cv(v) = \theta\{x\}$, since:
  \begin{gather*}
    \theta\{x\} = (\Resolver(\sigma') \circ [x |-> \cv(v)])(\{x\}) = \Resolver(\sigma')(\cv(v)) \\
    \theta\cv(v) = \Resolver(\sigma')(\cv(v)) \quad(\text{since } x \not\in \cv(v))
  \end{gather*}

\end{proof}

\subsubsection{Core lemmas}

\begin{lemma}[Program authority preservation]
  Let $\Psi[t]$ be a well-typed program such that $\Psi[t] --> \Psi[t']$.
  Then $\cv(t') \subseteq \cv(t)$.
\end{lemma}

\begin{proof}
  By inspection of the reduction rule used.

  \begin{proofcases}
    \pcase{\eruleref{apply}}
    Then $t = \storeof{\evalof{x\,y}}$ and $t' = \storeof{\evalof{[z |-> y]s}}$.
    Note that we have both $\cv(t) = \Resolver(\sigma)(\cv(e) \cup \cv(x\,y))$
    and $\cv(t') = \Resolver(\sigma)(\cv(e) \cup \cv([z |-> y]s))$.

    If we have $x \in \dom \Psi$, then $\Psi(x) = \lambda(z:U)\,s$.
    By definition of platform, the lambda is closed and we have $\fv(s) \subseteq \{z\}$,
    which in turn means that $\cv([z |-> y]s) \subseteq \{y\} \subseteq \cv(x\,y)$.

    Otherwise, we have $x \in \dom \sigma$ and $\sigma(x) = \lambda(z:U)\,s$.
    By definition of $\cv$ we have $\Resolve(\sigma, \cv(\lambda(z:U)\,s) \cup \{y\}) \subseteq \cv(\storeof{\evalof{x\,y}})$.
    Since $\cv(\lambda(z:U)\,s) \cup \{y\} \subseteq \cv([z |-> y]s)$, we can conclude.

    \pcase{\eruleref{tapply}}
    Analogous reasoning.

    \pcase{\eruleref{open}}
    Then $t = \storeof{\evalof{\Unboxed{C}{x}}}$ and $t' = \storeof{\evalof{z}}$.
    We must have $x \in \dom \sigma$ and $\sigma(x) = \Boxed z$, since all values bound in a platform must be closed and a box form cannot be closed.
    Since $\Psi[t]$ is a well-typed program, there must exist some $\Gamma,\Delta$ such that $\Gamma$ is a platform environment and $|- \Psi[\sigma] \sim \Gamma,\Delta $.

    If $z \in \dom \sigma$, then by Lemma~\ref{lm:inv-var-lookup} and Lemma~\ref{lm:inv-box} we have $\Gamma,\Delta |- z : C\,S_z$ for some $S_z$.
    By straightforward induction on the typing derivation, we then must have $\Gamma,\Delta |- \{z\} <: C$.
    Then by Lemma~\ref{lm:resolve-sc-relation} we have $\Resolver(\sigma)(\{z\}) \subseteq \Resolver(\sigma)(C)$,
    which lets us conclude by an argument similar to the \eruleref{apply} case.

    Otherwise, $z \in \dom \Psi$.
    Here we also have $\Gamma,\Delta |- \{z\} <: C$, which implies we must have $z \in C$,
    so we have $\cv(z) \subseteq \cv(\Unboxed C x)$ and can conclude by a similar argument as in the \eruleref{apply} case.

    \pcase{\eruleref{rename}, \eruleref{lift}}
    The lemma is clearly true since these rules only shift subterms of $t$ to create $t'$.
  \end{proofcases}
\end{proof}

\begin{lemma}[Single-step program trace prediction]
  Let $\Psi[t]$ be a well-typed program such that $\Psi[t] --> \Psi[t']$.
  Then the primitive capabilities used during this reduction are a subset of $\cv(t)$.
\end{lemma}

\begin{proof}
  By inspection of the reduction rule used.

  \begin{proofcases}
    \pcase{\eruleref{apply}}
    Then $t = \storeof{\evalof{x\,y}}$.
    If $x \in \dom \sigma$, the lemma vacuously holds.
    Otherwise, $x \in \dom \Psi \setminus \dom \sigma$.
    From the definition of $\cv$, we have $\{x\} \setminus \dom \sigma \subseteq \cv(t)$.
    Since $x$ is bound in $\Psi$, we then have $x \in \cv(t)$, which concludes.

    \pcase{\eruleref{tapply}}
    Analogous reasoning.

    \pcase{\eruleref{open}, \eruleref{rename}, \eruleref{lift}}
    Hold vacuously, since no capabilities are used by reducing using these rules.
  \end{proofcases}
\end{proof}

\begin{lemma}[Program trace prediction]
  Let $\Psi[t]$ be a well-typed program such that $\Psi[t] -->^{*} \Psi[t']$.
  Then the primitive capabilities used during this reduction are a subset of $\cv(t)$.
\end{lemma}
\begin{proof}
  By IH, Single-step program trace prediction and authority preservation.
\end{proof}

\subsection {Avoidance}
Here, we restate Lemma~\ref{lemma:avoiding-let-types-exist} and prove it.

\begin{lemma} \label{lemma:avoiding-let-types-exist:proof}
  Consider a term $\Let x s t$ in an environment $\Gamma$ such that $E \ts s : \Capt{C_s}{U}$
  is the most specific typing for $s$ in $\Gamma$ and $\Gamma, x : \Capt{C_s}{U} \ts t : T$ is the most specific
  typing for $t$ in the context of the body of the let, namely $\Gamma, x : \Capt {C_s}{U}$.  Let $T'$
  be constructed from $T$ by replacing $x$ with $C_s$ in covariant capture set positions and by replacing $x$
  with the empty set in contravariant capture set positions.  Then for every type $V$ avoiding $x$
  such that $\Gamma, x : \Capt{C_s}{S} \ts T \sub V$, $\Gamma \ts T' \sub V$. 
\end{lemma}
\begin{proof}
  We will construct a subtyping derivation showing that $T' \sub V$.
  Proceed by structural induction on the subtyping derivation for $T \sub V$.
  Since $T'$ has the same structure as $T$, most of the subtyping derivation
  carries over directly except for the subcapturing constraints in
  \ruleref{capt}.
  
  In this case, in covariant positions, 
  whenever we have $ C_T \sub C_V$
  for a capture set $C_T$ from $T$ and a capture set $C_V$ from $V$,
  we need to show that
  that $\ts C_T[x \to C_s] \sub C_V$.
  Conversely,
  in contravariant positions, whenever we have $C_V \sub C_T$, we need to show that $C_V \sub C_T[x \to \{\}]$.  
  For the covariant case, since $x \in C_T$ but not in $C_V$, by inverting the subcapturing relation $C_T \sub C_V$, we obtain $C_s \sub C_V$.
  Hence $C_T[x \to C_s] \sub C_V$, as desired.  

  The more difficult case is the contravariant case, when we have $C_V \sub C_T$.  Here, however,
  we have that $C_V \sub C_T[x \to \{\}]$ by structural induction on the subcapturing derivation as $x$ never
  occurs on the left hand side of the subcapturing relation as $V$ avoids $x$. 
\end{proof}